%% file: main.tex
\documentclass[runningheads]{llncs}
\usepackage[paperheight=235mm, paperwidth=155mm,textwidth=12.2cm,textheight=19.3cm]{geometry}
\usepackage{bbding}
\usepackage{graphicx}   
\makeatletter
\RequirePackage[bookmarks,unicode,colorlinks=true]{hyperref}%
   \def\@citecolor{blue}%
   \def\@urlcolor{blue}%
   \def\@linkcolor{blue}%

\def\orcidID#1{\smash{\href{http://orcid.org/#1}{\protect\raisebox{-1.25pt}{\protect\includegraphics{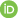}}}}}
\makeatother

\usepackage{bussproofs}
\usepackage{proof}
\usepackage{amsmath}
\usepackage{amssymb}
\usepackage{latexsym}  
\usepackage{epsf}
\usepackage{bbold}
\usepackage{stmaryrd}

\usepackage{pstricks}
\usepackage{amscd}

\usepackage{url} 
\usepackage{hyperref}

\input{macros}
\renewcommand{\HT}[1]{#1}
\renewcommand{\UB}[1]{#1}
\renewcommand{\mps}[1]{} 

\date{}

\begin{document}

\title{Extracting total Amb programs from proofs}         
\author{
Ulrich Berger\inst{1}
\orcidID{0000-0002-7677-3582} \and
Hideki Tsuiki\inst{2}
\orcidID{0000-0003-0854-948X}}



\authorrunning{U. Berger and H. Tsuiki}
%
\institute{Swansea University, Swansea, UK \\ 
\email{u.berger@swansea.ac.uk} 
\and
Kyoto University, Kyoto, Japan \\ 
\email{tsuiki@i.h.kyoto-u.ac.jp}} 

\maketitle              




\begin{abstract}
\input{abstract}
\end{abstract}

\section{Introduction}
\label{sec-introduction}
\input{sec-introduction}

\section{Denotational semantics of globally angelic choice }
\label{sec-ang}

\input{sec-ang}
\section{Operational semantics}
\label{sec-ops}
\input{sec-ops}


\section{CFP (Concurrent Fixed Point Logic)} 
\label{sec-cfp}
\input{sec-cfp}

\section{Program extraction}
\label{sec-pe}
\input{sec-pe} 




\section{Application}
\label{sec-gray}
\input{sec-gray} 

\section{Implementation}
\label{sec-experiments}
\input{sec-experiments}
\section{Conclusion}
\label{sec-conclusion}
\input{sec-conclusion}

\subsubsection*{Acknowledgements}
This work was supported by 
IRSES Nr.~612638 CORCON and Nr.~294962 COMPUTAL of
  the European Commission, the JSPS Core-to-Core Program, A. Advanced
  research Networks and JSPS KAKENHI  
15K00015 as well as 
  the Marie Curie RISE project CID (H2020-MSCA-RISE-2016-731143).





\bibliography{../../bibandmac/refs}

\vfill




\newpage

\section*{Appendix}

\appendix

\section{Proofs}
\label{appendix-proofs}
\input{appendix-proofs}

\section{Implementation}
\label{Sec:appendix-program}
\input{appendix-program}

\end{document}

%% file: macros.tex
\newcommand{\mps}[1]{}
\newcommand{\HT}[1]{{#1}}
\newcommand{\UB}[1]{{#1}}

\newcommand{\NN}{\mathbf{N}}

\newcommand{\D}{\mathbf{D}}
\newcommand{\E}{\mathbf{E}}

\renewcommand{\C}{\mathbf{S}}
\renewcommand{\G}{\mathbf{G}}

\newcommand{\gtos}{\mathsf{gtos}}
\newcommand{\nh}{\mathsf{nh}}
\newcommand{\nnot}{\mathsf{not}}

\newcommand{\onedigit}{\mathsf{onedigit}}

\newcommand{\gscomp}{\mathsf{gscomp}}
\newcommand{\Amb}{\mathbf{Amb}}

\newcommand{\Am}{\mathbf{A}}

\newcommand{\amb}{\mathbf{amb}}

\newcommand{\Fun}{\mathbf{Fun}}
 
\newcommand{\Left}{\mathbf{Left}}

\newcommand{\Nil}{\mathbf{Nil}}
\newcommand{\Pair}{\mathbf{Pair}} 

\newcommand{\Right}{\mathbf{Right}}

\newcommand{\SD}{\mathbf{SD}}
\newcommand{\Set}{{\mathbf{\downdownarrows}}}

\newcommand{\False}{\mathbf{False}}

\newcommand{\tent}{\mathbf{t}}

\newcommand{\caseof}[1]{\mathbf{case}\, #1\, \mathbf{of}\, }

\newcommand{\botexp}{\mathbf{\bot}}

\newcommand{\leftright}{\mathsf{leftright}}
\newcommand{\mapamb}{\mathsf{mapamb}}
\newcommand{\aup}{\mathsf{up}}
\newcommand{\adown}{\mathsf{down}}

\newcommand{\re}{\mathbf{r}}

\newcommand{\defined}[1]{{#1}\ne \bot}  
\newcommand{\eqdef}{\stackrel{\mathrm{Def}}{=}}
\newcommand{\eqmu}{\stackrel{\mu}{=}}
\newcommand{\eqnu}{\stackrel{\nu}{=}}
\newcommand{\eqmunu}{\stackrel{\munu}{=}}
\newcommand{\eqrec}{\stackrel{\mathrm{rec}}{=}}

\newcommand{\strictapp}[2]{#1{\downarrow}#2}

\newcommand{\ire}[2]{#1\,\mathbf{r}\,#2}

\newcommand{\inl}[1]{\Left(#1)}
\newcommand{\inr}[1]{\Right(#1)}

\newcommand{\all}[1]{\forall #1\,}
\newcommand{\ex}[1]{\exists #1\,}

\newcommand{\IFP}{\mathrm{IFP}}  
\newcommand{\CFP}{\mathrm{CFP}}  
\newcommand{\RIFP}{\mathrm{RIFP}}  
\newcommand{\RCFP}{\mathrm{RCFP}}  

\newcommand{\rec}{\mathbf{rec}}

\newcommand{\tfix}[2]{\mathbf{fix}\,#1\,.\,#2}

\newcommand{\ssp}{\rightsquigarrow}
\newcommand{\newprintp}{\overset{\mathrm{p}}{\ssp}}
\newcommand{\newprintptr}{\mathbin{\stackrel{\mathrm{p}}{\ssp}\kern-.25em{}^*}}
\newcommand{\newprintc}{\overset{\mathrm{c}}{\ssp}}

\newcommand{\valu}[2]{\llbracket{#1}\rrbracket#2}
\newcommand{\val}[1]{\llbracket{#1}\rrbracket}   
\newcommand{\rk}{\mathbf{rk}}

\newcommand{\ddata}{\mathrm{data}}

\newcommand{\rea}{\mathbf{R}}
\newcommand{\reah}{\mathbf{H}}

\newcommand{\reali}[1]{\tilde{#1}}

\newcommand{\seq}{\mathbf{seq}}

\newcommand{\CL}{\mathbf{CL}}
\newcommand{\IND}{\mathbf{IND}}

\newcommand{\COCL}{\mathbf{COCL}}
\newcommand{\COIND}{\mathbf{COIND}}

\newcommand{\mon}{\mathsf{mon}}

\newcommand{\mycomment}[1]{}

\newcommand{\allex}{\Diamond}
\newcommand{\munu}{\Box}

\newcommand{\pcv}[1]{\hat{#1}}

\newcommand{\one}{\mathbf{1}}

\newcommand{\sgb}{\mathrm{sgb}}

\newcommand{\ax}{\mathcal{A}}

\newcommand{\subdom}[1]{#1 \lhd D}

\newcommand{\tval}[2]{D^{#2}_{#1}}
\newcommand{\ftyp}[2]{#1 \Rightarrow #2}

\newcommand{\tri}{\mathbf{3}}
\newcommand{\bool}{\mathbf{2}}
\newcommand{\nat}{\mathbf{nat}}

\newcommand{\stream}[1]{#1^\omega}

\newcommand{\ConSD}{\mathbf{ConSD}}
\newcommand{\conSD}{\mathsf{conSD}}

\newcommand{\rt}[2]{#2|_{#1}}
\newcommand{\rtp}[2]{#2|'_{#1}}

%% file: abstract.tex
%
%

We present a logical system CFP (Concurrent Fixed Point Logic) that supports the extraction of nondeterministic and concurrent programs that are provably total and correct. CFP is an intuitionistic first-order logic with inductive and coinductive definitions extended by two propositional operators, $\rt{A}{B}$ (restriction, a strengthening of implication) and $\Set(B)$ (total concurrency). The source of the extraction are formal CFP proofs, the target is a lambda calculus with constructors and recursion extended by a constructor Amb (for McCarthy's amb) which is interpreted operationally as globally angelic choice and is used to implement nondeterminism and concurrency. The correctness of extracted programs is proven via an intermediate domain-theoretic denotational semantics. We demonstrate the usefulness of our system by extracting a nondeterministic program that translates infinite Gray code into the signed digit representation. A noteworthy feature of our system is that the proof rules for restriction and concurrency involve variants of the classical law of excluded middle that would not be interpretable computationally without Amb. 

%% file: sec-introduction.tex
Nondeterministic bottom-avoiding choice is an important and useful idea. 
With the wide-spread use of hardware that supports parallel computation,
it has the possibility to speed up practical computation and, at the same time,
it is related to computation over mathematical structures like real 
numbers~\cite{Escardo96,Tsuiki02}.
On the other hand, it is not easy to apply 
theoretical tools like denotational semantics to nondeterministic bottom-avoiding choice~\cite{HughesO89,Levy07} and guaranteeing 
correctness and totality of such programs through logical systems is a difficult task.
%

To explain the subtleness of the problem, let us start with an example.
Suppose that $M$ and $N$ are partial programs that,  
under the conditions $A$ and $\neg A$, respectively, 
are guaranteed to terminate and produce values satisfying specification $B$. 
Then, by executing $M$ and $N$ in parallel 
and taking the result obtained first, we should always obtain a result satisfying $B$. 
This kind of bottom-avoiding nondeterministic program
is known as \emph{McCarthy's amb (ambiguous) operator} \cite{McCarthy1963}, and 
we denote such a program by  $\Amb(M, N)$.
$\Amb$ is called the angelic choice operator and
is usually studied as one of the three 
nondeterministic choice operators (the other two are erratic choice and demonic choice).
On the other hand, 
we are interested in this operator not only from a theoretical 
point of view but also from the way it behaves as a
concurrent program running on a parallel execution mechanism.  

If one tries to formalize this idea naively, one will face some 
obstacles.  
Let $\ire{M}{B}$ (``$M$ realizes $B$'') denote the fact 
that a program $M$ satisfies a specification $B$ and 
let $\Set(B)$ be the specification that can be satisfied by 
a concurrent program of the form $\Amb(M, N)$ that  always terminates and 
produces a value satisfying $B$.  
Then, the above inference could be written as
\[
  \infer[\hbox{}]{
  \ire{\Amb(M, N)} \Set(B)
}{
A \to (\ire{M}{B})  \ \ \ \     \neg A \to (\ire{N}{B}) 
}
\]
However, this inference is not sound for the following reason.
Suppose that $A$ does not hold, that is, $\neg A$ holds. 
Then, the execution of $N$ will produce a
value 
satisfying $B$. But the execution of $M$ may terminate as well, and
with a data that does not satisfy $B$ since there is no condition on $M$
if $A$ does not hold.
Therefore,  if $M$ terminates first in the execution of $\Amb(M, N)$, 
then we obtain a result that may not satisfy $B$.

To amend this problem, we add a new operator $\rt{A}{B}$ 
(pronounced ``$B$ restricted to $A$'')
and consider the rule
\begin{equation}\label{eq0}
  \infer[\hbox{}]{
  \ire{\Amb(M, N)} \Set(B)
}{
\ire{M}{(\rt{A}{B})}\ \ \ \  \ire{N}{(\rt{\neg A}{B})}}
\end{equation}

Intuitively, $\ire{M}{(\rt{A}{B})}$ means  two things:
(1) $M$ terminates if $A$ holds, and
(2) if 
$M$ terminates, then the result satisfies $B$
even for the case $A$ does not hold.
As we will see in Sect.~\ref{sub-conc}, 
the above rule is derivable in classical logic
and can therefore be used to prove total correctness of Amb programs.

In this paper, we go 
a step further and
introduce a logical system $\CFP$ whose formulas
can be interpreted as specifications of 
nondeterministic programs
although they do not talk about programs explicitly.
$\CFP$ is defined by adding the two logical operators $\rt{A}{B}$ and 
$\Set(B)$ to the system $\IFP$, 
a logic for program extraction~\cite{IFP} 
(see also \cite{Berger11,SeisenBerger12,BergerPetrovska18}). 
A related approach has been developed in the
 proof system Minlog~\cite{SchwichtenbergMinlog06,BergerMiyamotoSchwichtenbergSeisenberger11,SchwichtenbergWainer12}.
$\IFP$ supports the extraction of lazy functional programs 
from inductive/coinductive proofs in intuitionistic first-order logic.
It has a prototype implementation in Haskell,
called Prawf \cite{DBLP:conf/cie/0001PT20}.

We show that from a $\CFP$-proof of a formula, both
a program and a proof that the program
satisfies the specification can be extracted
(Soundness theorem, Theorem \ref{thm-soundnessI}). 
For example, in $\CFP$ we have the rule 
 \begin{equation}\label{eq00}
   \infer[\hbox{(Conc-lem)}]{
   \Set(B)
 }{
 \rt{A}{B}\ \ \ \    \rt{\neg A}{B}}
  \end{equation}
which is realized by the program 
  $\lambda a. \lambda b. \Amb(a, b)$,
  and whose correctness is expressed by the rule (\ref{eq0}).
Programs extracted from $\CFP$ proofs can be  
executed in Haskell,
implementing $\Amb$
with  
the concurrent Haskell package.

Compared with program verification, the
extraction approach has the benefit that 
(a) the proofs 
programs are extracted from 
take place in a formal system
that is of a very high level of abstraction and therefore is simpler and
easier to use than a logic that formalizes concurrent programs
(in particular, programs do not have to be written manually at all);
(b) not only the complete extracted program is proven correct but also
all its sub-programs come with their specifications and correctness proofs
since these correspond to sub-proofs. This makes it easier to locally
modify programs without the danger of compromising overall correctness.

As an application,
we extract a nondeterministic 
program that converts 
infinite Gray code to signed digit representation, where 
infinite Gray code is a 
coding of real numbers by partial digit streams
that are allowed to contain a $\bot$, that is, a digit
whose computation does not terminate~\cite{Gianantonio99,Tsuiki02}. 
Partiality and multi-valuedness are common phenomena in computable analysis 
and exact real number computation~\cite{Weihrauch00,LUCKHARDT1977321}.
This case study connects these two aspects through a nondeterministic and
concurrent program whose correctness is guaranteed by a CFP-proof. 
\HT{The extracted Haskell programs are available in the repository~\cite{githubUB}
and are described in Appendix~\ref{Sec:appendix-program}.}


Organization of the paper: 
In Sects.~\ref{sec-ang} and~\ref{sec-ops}
we present the denotational and operational semantics of a functional 
language with $\Amb$ and prove that they match 
(Thms.~\ref{thm:data} and~\ref{thm:dataconv}).
Sects.~\ref{sec-cfp} and~\ref{sec-pe} describe the formal system $\CFP$ 
and its realizability interpretation 
which our program extraction method is based on 
(Thms.~\ref{thm-soundnessI} and~\ref{thm-soundnessII}).
In Sect.~\ref{sec-gray} we extract
a concurrent program that converts 
representation of real numbers 
and study its behaviour in Sect.~\ref{sec-experiments}.
%
%
%
%
\HT{Most proofs, unless very short, are omitted do to space limitation.
Full proofs of the main results can be found in the 
Appendix~\ref{appendix-proofs}. 
}

%% file: sec-ang.tex
In \cite{McCarthy1963}, 
McCarthy defined 
the ambiguity operator $\amb$ as
\[
\amb(x, y) = \left\{ \begin{array}{ll} x  &(x \ne \bot)\\
                       y & (y \ne \bot) \\
                       \bot & (x =  y = \bot)
                              \end{array}\right.
\]                              
where $\bot$ means `undefined' and $x$ and $y$ are taken nondeterministically 
when both $x$ and $y$ are not $\bot$.
This is called \emph{locally} angelic nondeterministic choice  
since convergence is chosen over divergence for each local call for 
the computation of $\amb(x, y)$.
It can be implemented by executing both of the arguments
in parallel and taking the result obtained first.  
Despite being a simple construction,
$\amb$ is known to have a lot of expressive power, 
and many constructions of nondeterministic and parallel computation 
such as erratic choice, countable choice (random assignment), 
and `parallel or' can be encoded through it \cite{LassenMoran99}.
These multifarious aspects of the operator $\amb$ are reflected by 
the difficulty of its mathematical treatment in denotational semantics.
For example, 
$\amb$ is not monotonic when interpreted over powerdomains
with the Egli-Milner order~\cite{Broy1986}.

On the other hand, one can consider an interpretation of $\amb$ as  
\emph{globally} angelic choice, where  an argument of $\amb$ 
is chosen so that the whole ambient computation converges, 
if convergence 
is possible at all~\cite{ClingerHalpern85,SondergardSestoft92}.
Since globally angelic choice is not defined compositionally,  
it is not easy to integrate it into a design of a programming 
language with clear denotational semantics.
However, 
it
can be easily implemented by 
running the whole computation for both
of the arguments of $\amb$ in parallel and taking the 
result obtained first.
%
%
Denotationally, globally angelic choice 
can be modelled by
the Hoare powerdomain construction. 
However, this would not be suitable for analyzing total correctness
  because the ordering of the Hoare powerdomain does 
not discriminate $X$ and $X \cup \{\bot\}$~\cite{HughesM92,HughesO89}.
Instead, we consider a 
two-staged approach (see~Sect.~\ref{sub-denot}).  
%

The difference between 
the locally and the globally angelic interpretation
%
of $\amb$ 
is highlighted
by the fact that the 
former 
does not commute with function application.
For example, if 
$f(0) = 0$ but $f(1)$ diverges,
then $\amb(f(0), f(1))$ will always terminate with the value $0$, whereas
$f(\amb(0, 1))$ may return 0 or diverge.
On the other hand, 
the latter term will always return $0$ 
if $\amb$ is implemented with a globally angelic semantics.
As suggested in \cite{ClingerHalpern85}, 
we use this commutation property to realize the globally angelic semantics.

\subsection{Programs and types}
\label{Sec:2.1}
%
%
Our target language for program extraction
is an untyped lambda calculus with recursion operator and
constructors as in~\cite{IFP}, but extended by an additional 
constructor $\Amb$ that corresponds to globally angelic version
of McCarthy's $\amb$.
This could be easily generalized to an $\Amb$ operator of any arity $\ge 2$.
%
%
%
%
\label{sub-prog}
%
%
%
%
\begin{align*}
&  \mathit{Programs} \owns M,N,L,P, Q,R :: = 
  a, b, \ldots, f, g \ \ \text{(program variables)}\\
&\quad  |\  \lambda a.\,M 
\ | \ M\,N 
\ | \ \strictapp{M}{N} 
\ | \ \rec\,M \ | \  \botexp \\
&\quad |\    \Nil\ | \ \Left(M)\ | \ \Right(M)\ | \ \Pair(M,N) |\ \Amb(M,N)\\
&\quad |\  \caseof{M}\{\Left(a) \to L; \Right(b) \to R\} \\ 
&\quad |\  \caseof{M}\{\Pair(a,b) \to N\} \\
&\quad |\  \caseof{M}\{\Amb(a,b) \to N\}  
\end{align*}
%
%
Denotationally, $\Amb$ is just another pairing operator. 
Its interpretation as globally angelic choice 
will come to effect only through its operational semantics.
Though essentially a call-by-name language, 
it also has strict application $\strictapp{M}{N}$, 
needed for realizing the rules for restriction and the concurrency
operator.
%

We use $a, \ldots, g$ for program variables to distinguish them from the variables
$x, y, z$ of 
the logical system 
CFP (Sect.~\ref{sec-cfp}). 
%
%
%
$\Nil, \Left, \Right, \Pair, \Amb$ are called \emph{constructors}.
Constructors different from $\Amb$ are called \emph{data constructors}.
$\mathrm{C_d}$ denotes the set of data constructors.
$\strictapp{\Left}{M}$ stands for $\strictapp{(\lambda a.\Left(a))}{M}$, etc.,
and we sometimes write $\Left$ and $\Right$ for $\Left(\Nil)$ and 
$\Right(\Nil)$.  
%
Natural numbers are
 encoded as $0 \eqdef \Left$,  $1 \eqdef \Right(\Left)$, and so on.

%
Although programs are untyped, programs extracted from
proofs will be typable by 
the following system of 
simple recursive types:
\[
Types \ni \rho, \sigma ::=  \alpha\ (\hbox{type variables})
                         \mid \one  
                         \mid \rho \times \sigma
                         \mid \rho + \sigma
                         \mid \ftyp{\rho}{\sigma}
                         \mid \tfix{\alpha}{\rho}
                         \mid \Am(\rho) 
\]
%
%
%
Here, $\Am(\rho)$ is the type of
programs which, if they terminate (see Sect.~\ref{sec-ops}), 
reduce to a form
$\Amb(M, N)$ with $M,N\!:\!\rho$.
The formation of $\tfix{\alpha}{\rho}$ has the side conditions
that 
%
$\alpha$ occurs freely in $\rho$,
$\rho$ is strictly positive in $\alpha$ (that is, there is no free 
occurrence of $\alpha$ in $\rho$ which is in the left part of a function type), 
and not of the form $\alpha$ or $\Am(\alpha)$.
These conditions ensure, among other things, 
that the type transformer $\alpha \mapsto \rho$
has a unique fixed point, which is taken as the semantics of $\tfix{\alpha}{\rho}$ 
 (see below).
We require in $\Am(\rho)$ that $\rho$ \UB{is neither a variable nor of} 
the form 
$\tfix{\alpha_1}{\ldots{\tfix{\alpha_n}{\Am(\rho')}}}$ ($n \ge 0$).
\mps{\HT{I wonder if we need it.}
\UB{I think it's nicer with the extra restriction to have closure under subterms. Wording slightly changed.}}
This enables the interpretation of $\Amb$ as a bottom-avoiding choice operator
(see the explanation below Corollary~\ref{cor:ddatabot}).
We call types that satisfy all these conditions  \emph{regular}.
%
An example of a regular type is 
the type of lazy (partial) natural numbers, 
$\nat \eqdef \tfix{\alpha}{\one+\alpha}$.

\begin{figure}
\fbox{
\begin{minipage}{\textwidth}
\begin{center}
$\Gamma, a:\rho \vdash a:\rho$
\hspace{3em} 
$\Gamma \vdash \Nil:\one$
\hspace{3em}
$\Gamma \vdash \bot:\rho$
\hspace{3em}
\\[0.5em]
\AxiomC{$\Gamma\vdash M:\rho$}
             \UnaryInfC{$\Gamma \vdash \Left(M) : \rho + \sigma$}
            \DisplayProof 
\hspace{3em} 
\AxiomC{$\Gamma\vdash M:\sigma$}
             \UnaryInfC{$\Gamma \vdash \Right(M) : \rho + \sigma$}
            \DisplayProof \ \ \ \ 
\\[0.5em]
\AxiomC{$\Gamma\vdash M:\rho$}
\AxiomC{$\Gamma\vdash N:\sigma$}
             \BinaryInfC{$\Gamma \vdash \Pair(M,N) : \rho\times\sigma$}
            \DisplayProof 
\ \ \ \ \ \ \ \ 
\AxiomC{$\Gamma\vdash M:\rho$}
\AxiomC{$\Gamma\vdash N:\rho$}
             \BinaryInfC{$\Gamma \vdash \Amb(M,N) : \Am(\rho)$}
            \DisplayProof 
\\[0.5em]
\AxiomC{$\Gamma, a:\rho\vdash M:\sigma$}
             \UnaryInfC{$\Gamma \vdash \lambda a.\,M : \ftyp{\rho}{\sigma}$}
            \DisplayProof 
\ \ \ \ \ \ \ \ \ \ 
\AxiomC{$\Gamma, a:\rho\vdash M\,a:\rho$}
             \UnaryInfC{$\Gamma \vdash \rec\,M : \rho$}
            \DisplayProof 
{($a$ not free in $M$)}
\\[0.5em]
\AxiomC{$\Gamma\vdash M:\ftyp{\rho}{\sigma}$}
\AxiomC{$\Gamma\vdash N:\rho$}
             \BinaryInfC{$\Gamma \vdash M\,N : \sigma$}
            \DisplayProof \ \ \ \ 
\hspace{3em} 
\AxiomC{$\Gamma\vdash M:\ftyp{\rho}{\sigma}$}
\AxiomC{$\Gamma\vdash N:\rho$}
             \BinaryInfC{$\Gamma \vdash \strictapp{M}{N} : \sigma$}
            \DisplayProof \ \ \ \ 
\\[0.5em]
\AxiomC{$\Gamma \vdash M : \rho[\tfix{\alpha}{\rho}/\alpha]$}
\RightLabel{{}}
             \UnaryInfC{$\Gamma \vdash M : \tfix{\alpha}{\rho}$}
            \DisplayProof 
\hspace{3em} 
\AxiomC{$\Gamma \vdash M : \tfix{\alpha}{\rho}$}
\RightLabel{{}}
             \UnaryInfC{$\Gamma \vdash M : \rho[\tfix{\alpha}{\rho}/\alpha]$}
            \DisplayProof 
\\[1em]
%
\AxiomC{$\Gamma \vdash M:\rho + \sigma$\ \ \ \ 
$\Gamma, a:\rho \vdash L:\tau$\ \ \ \ 
$\Gamma, b:\sigma \vdash R:\tau$}
\UnaryInfC{$\Gamma\vdash \caseof{M} \{\Left(a) \to L; \Right(b) \to R \} :\tau$}
            \DisplayProof \ \ \ \ 
\\[1em]
\AxiomC{$\Gamma \vdash M:\rho\times\sigma$\ \ \ \ 
$\Gamma, a:\rho, b:\sigma \vdash N:\tau$}
\UnaryInfC{$\Gamma\vdash \caseof{M} \{\Pair(a, b) \to N\} :\tau$}
            \DisplayProof 
\ \ 
\AxiomC{$\Gamma \vdash M:\Am(\rho)$\ \ \ 
$\Gamma, a,b:\rho \vdash N:\tau$}
\UnaryInfC{$\Gamma\vdash \caseof{M} \{\Amb(a, b) \to N\} :\tau$}
            \DisplayProof 
%
%
\end{center}
\end{minipage}
}
\caption{Typing rules}
\label{fig-typing}
\end{figure}

The typing rules are listed in Fig.~\ref{fig-typing}.
They are valid w.r.t.~the denotational semantics given in Sect.~\ref{sub-denot} 
and extend the rules given in~\cite{IFP}. 
%
Recursive types are equirecursive \cite{Pierce:2002} in that
$M : \tfix{\alpha}{\rho}$ iff $M : \rho[\tfix{\alpha}{\rho}/\alpha]$.
%

As an example of a program consider
\begin{equation}
\label{eq-f}
f\eqdef\lambda a.\caseof{a} \{\Left(\_)\to\Left;
                                    \Right(\_)\to \bot\}
\end{equation}
which implements the function $f$ discussed earlier, i.e.,
$f\,0 = 0$ and $f\,1 = \bot$.
$f$ has type $\ftyp{\nat}{\nat}$. 
Since 
$\Amb(0,1)$ has type $\Am(\nat)$,
the application $f\,\Amb(0,1)$ 
is not well-typed. 
Instead, we consider
$
\mapamb\ f\ \Amb(0, 1)
$
where 
$\mapamb : (\rho \to \sigma) \to \Am(\rho) \to \Am(\sigma)$ is defined as
\begin{eqnarray*}
\mapamb\eqdef \ \lambda f.\, \lambda c.\ \caseof{c}\,
\{\Amb(a,b) \to 
 \Amb(\strictapp{f}{a}, \strictapp{f}{b})\}
\end{eqnarray*}
%
%
%
%

This operator realizes the globally angelic semantics:
$\mapamb\ f\ \Amb(0, 1)$ is reduced to $\Amb(\strictapp{f}{0}, \strictapp{f}{1})$,
 and $\strictapp{f}{0}$ and $\strictapp{f}{1}$ (which are
the same as $f\ 0$ and $f\ 1$ since $0$ and $1$ are defined) 
are computed concurrently and 
the whole expression
is reduced to 0, using 
the operational semantics in Section \ref{sec-ops}.
In Sect.~\ref{sec-pe}, we will introduce a concurrent (or nondeterministic) 
version of Modus Ponens, (Conc-mp), 
which will automatically generate an application of $\mapamb$.
\subsection{Denotational semantics}
\label{sub-denot}
%
%
%
%
%
The denotational semantics has two phases: 
\emph{Phase~I} interprets programs in a Scott domain $D$
defined by the 
following
recursive domain equation
%
\[
D = (\Nil + \Left(D) + \Right(D) + \Pair(D\times D)
    + \Amb(D\times D) + \Fun(D\to D))_\bot \,.
\]
where $+$ and $\times$ denote separated sum and cartesian product, 
and the operation $\cdot_\bot$ adds a least element 
$\bot$~(\cite{GierzHofmannKeimelLawsonMisloveScott03} is a recommended 
reference for domain theory and the solution of domain equations).
%
%
%
%
%
A closed program $M$ denotes an element $\val{M}\in D$
as defined in Fig.~\ref{fig-semantics}.
Note that $\Amb$ is interpreted (like $\Pair$) as a simple pairing operator.

A type is interpreted as a subdomain, which is 
a subset of $D$
that is downward closed and closed under suprema of bounded subsets.
We use 
the following operations on subdomains: 
\begin{eqnarray*}%
(X+Y)_\bot &\eqdef& \{\Left(a)\mid a\in X\} \cup 
                           \{\Right(b) \mid b \in Y\} \cup \{\bot\}\\
(X\times Y)_\bot &\eqdef& \{\Pair(a,b) \mid a\in X, b\in Y \} \cup \{\bot\}\\
(\ftyp{X}{Y})_\bot &\eqdef& \{\Fun(f) \mid f:D\to D\hbox{ continuous, }
                           \forall a\in X (f(a) \in Y) \} \cup \{\bot\}.
\end{eqnarray*}
Through the semantics in Fig.~\ref{fig-semantics}, 
closed programs denote elements of $D$ and closed types denote subdomains of $D$
such that the typing rules (Fig.~\ref{fig-typing}) are sound.  

\begin{figure}
\fbox{\small
\begin{minipage}{\textwidth}
\vspace*{-0.2cm}
\begin{eqnarray*}
\valu{a}{\eta} &=& \eta(a)\\
\valu{\lambda a.\,M}{\eta} &=& \Fun(f)\quad  
  \hbox{where $f(d) = \valu{M}{\eta[a\mapsto d]}$}\\
\valu{M\,N}{\eta} &=& f(\valu{N}{\eta})\quad 
                        \hbox{if $\valu{M}{\eta}= \Fun(f)$}\\
\valu{\strictapp{M}{N}}{\eta} &=& f(\valu{N}{\eta})\quad 
                        \hbox{if $\valu{M}{\eta}= \Fun(f)$ and $\val{N}\neq\bot$}\\
\valu{\rec\,M}{\eta} &=& \hbox{the least fixed point of $f$ if $\valu{M}{\eta}=\Fun(f)$}\\
\valu{C(M_1,\ldots,M_k)}{\eta} &=& C(\valu{M_1}{\eta},\ldots,
                                       \valu{M_k}{\eta})
\quad \hbox{($C$ a constructor (including $\Amb$))}\\
\valu{\caseof{M}{\vec{Cl}\}}}{\eta} &=& 
       \valu{K}{\eta[\vec a \mapsto \vec d]} 
       \quad \hbox{if $\valu{M}{\eta} = C(\vec d)$ and $C(\vec a) \to K\in\vec{Cl}$}\\
%
%
\valu{M}{\eta} &=& \bot\ \ 
\hbox{in all other cases, 
in particular $\valu{\botexp}{\eta} = \bot$}
\end{eqnarray*}

\vspace*{-0.2cm}  

\hspace{1cm}${\eta}$ is an environment that assigns elements of $D$ to variables.
%
%
\begin{eqnarray*}
%
\tval{\alpha}{\zeta} &=& \zeta(\alpha),\qquad
\tval{\one}{\zeta} \ =  \ \{\Nil,\bot\},\\
\tval{\tfix{\alpha}{\rho}}{\zeta} &=& 
\bigcap\{\subdom{X} \mid \tval{\rho}{\zeta[\alpha \mapsto X]} \subseteq X \} 
\text{\qquad ($\subdom{X}$ means $X$ is a subdomain of $D$)}\\
\tval{\Am(\rho)}{\zeta} &=&
  \{\Amb(a, b) \mid a, b \in \tval{\rho}{\zeta}\} \cup \{\bot\}\\
\tval{\rho\diamond\sigma}{\zeta} &=& 
(\tval{\rho}{\zeta}\diamond\tval{\sigma}{\zeta})_{\bot}\quad 
(\diamond\in\{+,\times,\ftyp{}{}\})   
\end{eqnarray*}                           

\vspace*{-0.2cm}  

\hspace*{1cm}
${\zeta}$ is a type environment that assigns subdomains $D$ to type variables.
\end{minipage}
}
\caption{Denotational semantics of programs (Phase I) and types}
\label{fig-semantics}
\end{figure}

%
In \emph{Phase~II} we assign to every $a\in D$
a set $\ddata(a)\subseteq D$ 
that reveals the role of $\Amb$ as a choice operator. 
The relation `$d\in\ddata(a)$' 
is defined (coinductively) as the largest relation satisfying
\begin{align*}
d \in \ddata(a) \quad \eqnu \quad 
&(a = \Amb(a', b') \land a' \ne \bot \land d \in \ddata(a')) \ \lor \\
&(a = \Amb(a', b') \land b' \ne \bot \land  d \in \ddata(b')) \ \lor \\
&(a = \Amb(\bot, \bot) \land d = \bot) \ \lor \\  
&\bigvee_{C \in \mathrm{C_d}}
\left(a =  C(\vec{a'}) \land d = C(\vec{d'}) \land \bigwedge_i d'_i \in \ddata(a'_i)\right) \ \lor\\
&(a = \Fun(f)  \land d = a ) \ \lor (a = d = \bot)\,.
\end{align*} 
%
%
%
Now, every closed program $M$ 
denotes the set 
$\ddata(\val{M}) \subseteq D$ containing all possible globally 
angelic choices derived form its denotation in $D$.
For example, $\ddata(\Amb(0,1)) = \{0,1\}$ and, for $f$ as
defined in~(\ref{eq-f}), we have, as expected,
$\ddata(\mapamb\ f\ \Amb(0, 1)) = \ddata(\Amb(0,\bot)) = \{0\}$. 
%
In Sect.~\ref{sec-ops} we will define an operational semantics whose fair
execution sequences starting with a regular-typed program 
$M$ compute exactly the elements in $\ddata(\val{M})$.

\begin{example}\label{ex:random}
  Let $M = \rec\ \lambda a. \Amb(\Left(\Nil), \Right(a))$. $M$ is a closed program
of type
  $\tfix \alpha \Am(\one + \alpha)$.
  We have $\ddata(M) = \{0, 1, 2, \ldots\}$.
Thus, we can express countable choice (random assignment) with $\Amb$.
\end{example}

\begin{lemma}\label{lem:ddatabot}
If $a \in D$ belongs to a regular type, then the following are equivalent:
  (1) $a \in\{ \bot,\Amb(\bot, \bot)\}$;   
  (2) $\{\bot\} = \ddata(a)$;
  (3) $\bot \in \ddata(a)$.
\end{lemma}

%% file: sec-ops.tex
%
We define a small-step operational semantics 
that, in the limit, reduces each closed program $M$ nondeterministically
to an element in $\ddata(\val{M})$ (Thm.~\ref{thm:data}). 
If $M$ has a regular type, the converse holds as well: 
For every $d\in\ddata(\val{M})$ there exists
a reduction sequence for $M$ computing $d$ in the limit 
(Thm.~\ref{thm:dataconv}).
If $M$ denotes a compact data, then the limit is obtained 
after finitely many reductions.
%
%
%
In the following, all programs are assumed to be closed.

\subsection{Reduction to weak head normal form}
A program is called a \emph{weak head normal form (w.h.n.f.)} if it
 begins with a constructor 
(including $\Amb$), 
or has the form $\lambda a. M$.
We define inductively 
a 
small-step leftmost-outermost reduction 
relation $\ssp$ on 
programs
where $C$ ranges over constructors. 
%

\medbreak
\begin{itemize}
\setlength{\itemsep}{0.2cm}
\setlength{\itemindent}{0.4cm}
  %
\item[(s-i)] $(\lambda a.\,M)\ N \ssp M[N/a]$
%
\item[(s-ii)] \AxiomC{$M \ssp M'$}
            \UnaryInfC{$M\,N \ssp M'\,N$}
            \DisplayProof 
%
\item[(s-iii)]  $\strictapp{(\lambda a.\,M)}{N} \ssp M[N/a]$ \quad 
if $N$ is a w.h.n.f.
%
\item[(s-iv)] \AxiomC{$M \ssp M'$}
            \UnaryInfC{$\strictapp{M}{N} \ssp \strictapp{M'}{N}$}   
            \DisplayProof  \quad  
if $N$ is a w.h.n.f.
%
\item[(s-v)]  \AxiomC{$N \ssp N'$}
            \UnaryInfC{$\strictapp{M}{N} \ssp \strictapp{M}{N'}$}
            \DisplayProof 
%
\item[(s-vi)] $\rec\,M \ssp M\,(\rec\,M)$
%
\item[(s-vii)] $\caseof{C(\vec M)} \{\ldots;C(\vec b)\to N;\ldots\}\ssp  N[\vec M/\vec b]$
%
\item[(s-viii)] \AxiomC{$M \ssp M'$}
             \UnaryInfC{$\caseof{M}\{\vec{Cl}\}\ssp 
               \caseof{M'}\{\vec{Cl}\}$}
            \DisplayProof 
\item[(s-ix)] $M \ssp \bot$ \quad if $M$ is $\bot$-like (see below)
\end{itemize}
$\bot$-like programs are such that their syntactic 
forms immediately imply that
they denote $\bot$, more precisely they are of the form  $\bot$,
$C(\vec M)\,N$, $\strictapp{C(\vec M)}{N}$, 
and $\caseof{M}\,\{\ldots\}$ where $M$ is a lambda-abstraction or
of the form $C(\vec M)$ such that there is no clause in $\{\ldots\}$ 
which is of the form $C(\vec a) \to N$. 
%
W.h.n.f.s 
are never $\bot$-like, and the only typeable $\bot$-like
program is $\bot$.
%

\begin{lemma}\label{lem:ssp}
\begin{enumerate}
\item[(1)] $\ssp$ is deterministic (i.e., $M \ssp M'$ for at most one $M'$).
\item[(2)] $\ssp$ preserves the denotational semantics 
(i.e., 
$\val{M} = \val{M'}$ 
if $M \ssp M'$).
\item[(3)] $M$ is a $\ssp$-normal form iff $M$ is a  
w.h.n.f.  
\item[(4)] [Adequacy Lemma]\label{lem:ade}
  If $\val{M} \ne \bot$, then there is a 
w.h.n.f.~$V$ 
s.t.\ $M \ssp^* V$.
\end{enumerate}
\end{lemma}



%
%




\subsection{Making choices}
Next, we define the reduction relation $\newprintc$ (`c' for 'choice') 
that reduces 
arguments of $\Amb$ in parallel. 

\begin{itemize}
\setlength{\itemsep}{0.2cm}
\setlength{\itemindent}{0.4cm}
\item[(c-i)] \AxiomC{$M \ssp M' $} 
\UnaryInfC{$M \newprintc M'$}
\DisplayProof 

%
\item[(c-ii)] 
\AxiomC{$M_1 \ssp M_1'$}  
\UnaryInfC{$\Amb(M_1,M_2) \newprintc \Amb(M_1',M_2)$}
\DisplayProof 
\item[(c-ii')] 
\AxiomC{$M_2 \ssp M_2'$}  
\UnaryInfC{$\Amb(M_1,M_2) \newprintc \Amb(M_1,M_2')$}
\DisplayProof 
%
%


\item[(c-iii)] 
$\Amb(M_1, M_2) \newprintc M_1$  \ 
if $M_1$ is a w.h.n.f.


\item[(c-iii')] 
$\Amb(M_1, M_2) \newprintc M_2$  \ 
if $M_2$ is a w.h.n.f.

\end{itemize}


%
From this definition and Lemma \ref{lem:ssp}, it is immediate that
$M$ is a $\newprintc$-normal form iff $M$ is a
\emph{deterministic weak head normal form (d.w.h.n.f.)}, that is,
a w.h.n.f.\ that does not begin with $\Amb$.
Finally, we define a reduction relation $\newprintp$ 
that reduces arguments of data constructors in parallel.
%
\begin{itemize}
\setlength{\itemsep}{0.2cm}
\setlength{\itemindent}{0.4cm}
\item[(p-i)]
\AxiomC{$M \newprintc M' $} 
\UnaryInfC{$M \newprintp M'$}
\DisplayProof 
\item[(p-ii)]
\AxiomC{$M_i \newprintp M_i'$ $(i = 1,\ldots, k)$} 
\UnaryInfC{$C(M_1,\ldots,M_k) \newprintp C(M_1',\ldots,,M_k')$}
\DisplayProof 
($C\in\mathrm{C_d}$)
\\
\item[(p-iii)]
  $\lambda a.\,M \newprintp \lambda a.\,M$ 
\end{itemize}
%
Every (closed) 
program reduces under $\newprintp$
(easy proof by structural induction). 
For example, $\Nil\newprintp \Nil$ 
 by (p-ii). In the following, all $\newprintp$-reduction sequences are assumed 
to be infinite.




%
We call a $\newprintp$-reduction sequence 
\emph{unfair} if, intuitively, from some point on, one side of an 
$\Amb$ term is permanently reduced but not the other. 
More precisely, we inductively  define 
$M_1\newprintp M_2 \newprintp \ldots$ to be unfair if
\begin{itemize}
\item each $M_i$ is of the form $\Amb(L_i,R)$ (with fixed $R$) 
and $L_i \ssp L_{i+1}$, or  
%
\item each $M_i$ is of the form $\Amb(L,R_i)$ (with fixed $L$) and
$R_i \ssp R_{i+1}$, or   
%
\item each $M_i$ is of the form $C(N_{i,1},\ldots,N_{i,n})$ 
(with a fixed $n$-ary constructor $C$) and $N_{1,k} \newprintp N_{2,k}\newprintp \ldots$
is unfair for some $k$, or
\item the tail of the sequence, $M_2\newprintp M_3\ldots$, is unfair.
\end{itemize}
A $\newprintp$-reduction sequence is \emph{fair} if it is not unfair.

Intuitively,  
reduction by $\newprintp$ proceeds as follows:
A program $L$ is head reduced by $\ssp$ to a 
w.h.n.f.\ $L'$, 
and
if $L'$ is a data constructor term, all arguments are reduced in parallel by (p-ii).
If $L'$ has the form $\Amb(M, N)$, 
two concurrent threads 
are invoked for the reductions of $M$ and $N$ in parallel, 
and the one reduced to a 
w.h.n.f.\ first is used. 
Fairness corresponds to the fact that the `speed' of each thread is positive
which means, in particular, that no thread can block another.
Note that $\newprintc$ is not used for the reductions of 
$M$ and $N$ in (s-ii), (s-iv), (s-v) and (s-viii).
This means that  $\newprintc$ is applied only to the outermost redex.
Also, (c-ii) is 
defined through $\ssp$, not $\newprintc$,  
and thus no thread creates new threads.
This ability to limit the bound of threads was not available
in an earlier version of this language~\cite{BergerCSL16}
(see also the discussion in Sect.~\ref{sub-related}).

\subsection{Computational adequacy:
  Matching denotational and operational semantics}
We define 
${M}_{{D}} \in {D}$ by structural induction on programs:
%
\begin{align*}
{C(M_1,\ldots,M_k)}_{D} &= C({M_1}_{D},\ldots, {M_k}_{D})
& (C\in \mathrm{C_d})\\
{(\lambda a. M) }_{D} &= \val{\lambda a. M}   \\
{M}_{D} &= \bot  & \mbox{otherwise}
\end{align*}
Since clearly $M \newprintp N$ implies $M_{D} \sqsubseteq_{D} N_{D}$,
for every computation sequence $ M_0 \newprintp M_1 \newprintp  \ldots$,  
  the sequence $((M_i)_{D})_{i \in \NN}$ is increasing and therefore
  has a least upper bound in $D$. 
Intuitively, $M_{D}$ is the part of $M$ that has been fully evaluated to a data.




A {\em computation} of $M$ is \UB{an infinite} fair 
sequence 
$M = M_0 \newprintp M_1 \newprintp \ldots$.  
\begin{theorem}[Computational Adequacy: Soundness]
\label{thm:data}
For every computation  $M =  M_0 \newprintp M_1 \newprintp  \ldots$, 
$\sqcup_{i \in \NN} (M_i)_{D} \in \ddata(\val{M})$. 
\end{theorem}

The converse 
does not hold in general,  
i.e.\ 
$d \in \ddata(\val{M})$ does not necessarily imply
$d = \sqcup_{i \in \NN} ((M_i)_D)$  
for some computation of $M$.
%
For example, for 
$M \eqdef \rec\, \lambda\, a.\, \Amb(a,\bot)$
(for which $\val{M} = \val{\Amb(M,\bot)}$) one sees 
that $d \in \ddata(\val{M})$ 
for every $d \in D$ while $M \newprintp^* M$ and $M_{D} = \bot$.
But 
$M$ has the type $\tfix{\alpha}{\Am(\alpha)}$ which is not
regular (see Sect.~\ref{sub-prog}). 
For programs of a regular type, 
the converse of Thm.~\ref{thm:data} holds. 

\begin{theorem}[Computational Adequacy: Completeness]
\label{thm:dataconv}
If $M$ has a regular type, then
for every $d \in \ddata(\val{M})$, there is a computation 
$M = M_0 \newprintp M_1 \newprintp \ldots$ with
$d = \sqcup_{i \in \NN} ((M_i)_{D})$.
%
\end{theorem}
%
A computation $M = M_0 \newprintp M_1 \newprintp \ldots$ 
is {\em productive} if some $M_i$ is a deterministic w.h.n.f.
Clearly, this is the case iff 
$\sqcup_{i \in \NN} ((M_i)_{D})\neq\bot$.
Therefore,
by the Adequacy Theorem and Lemma~\ref{lem:ddatabot}: 
%
\begin{corollary}\label{cor:ddatabot}
For a program $M$ of regular type, the following 
are equivalent.
  \begin{enumerate}
   \item[(1)] One of the  computations of $M$ is productive.
   \item[(2)] All 
computations  of $M$ are productive.  \mps{\HT{Deleted infinite because computation means infinite in definition.}}
   \item[(3)] $\val{M}$ is neither $\bot$ nor $\Amb(\bot,\bot)$.
\end{enumerate}
\end{corollary}
%
The corollary 
does not hold without  the 
regularity
condition.
For example,  $M = \Amb(\Amb(\Nil,\Nil), \Amb(\bot, \bot))$  can be reduced to $M_1 = \Amb(\bot, \bot)$ and then
repeats $M_1$ forever,  whereas it can also be reduced to $\Nil$.
\HT{
McCarthy's $\amb$ operator is bottom-avoiding
in that when it can terminate, it always terminates.
\UB{Corollary~\ref{cor:ddatabot}}  guarantees a similar property for our 
globally angelic choice operator $\Amb$.
}

\mps{\HT{omitted the last paragraph because not so clear.}}

%% file: sec-cfp.tex
In \cite{IFP}, the system $\IFP$ (Intuitionistic Fixed Point Logic) was 
introduced.
$\IFP$ is an intuitionistic first-order logic with strictly positive 
inductive and coinductive definitions, from the proofs of which programs can be extracted.
%
$\CFP$ 
is obtained by adding to $\IFP$ two propositional
operators, $\rt{A}{B}$ and $\Set(B)$, that facilitate the extraction of 
nondeterministic and concurrent 
programs. 
%

\subsection{Syntax}
\label{sub-cfpsyntax}
$\CFP$ is defined relative to a
many-sorted first-order language. 
%
$\CFP$-formulas 
have the form 
$A \land B$,  $A \lor B$, $A \to B$,  
$\forall x\, A$, $\exists x\, A$, 
$s = t$ ($s$, $t$ terms of the same sort), 
$P(\vec t)$ (for a predicate $P$ and terms $\vec t$ of fitting arities),
as well as $\rt{A}{B}$ (restriction) and $\Set(B)$ 
(concurrency).
%
Predicates are either predicate constants (as given by the first-order language),
or predicate variables (denoted $X,Y,\ldots$), 
or comprehensions $\lambda\vec x\,A$ (where $A$ is a formula and $\vec x$ 
is a tuple of first-order variables), 
or fixed points $\mu(\Phi)$ and $\nu(\Phi)$
(least fixed point aka inductive predicate 
and greatest fixed point aka coinductive predicate) where $\Phi$ is a
strictly positive (s.p.) operator. Operators are of the form
$\lambda X\,Q$ where $X$ is a predicate variable and $Q$ is a predicate  
of the same arity as $X$. 
$\lambda X\,Q$ is s.p.\ if 
every free occurrence of $X$ in $Q$ 
is at a strictly positive position,
that is, at a position that is not in
the left part of an implication.
We identify $(\lambda\vec x\,A)(\vec t)$ with $A[\vec t/\vec x]$ where
$[\vec t/\vec x]$ means capture avoiding substitution.

The following syntactic properties of expressions (i.e., formulas, predicates and operators)
will be important. A \emph{Harrop} expression is one that 
contains at strictly positive positions neither free predicate variables nor 
disjunctions ($\lor$) nor restrictions  ($\rt{}{}$)
nor concurrency ($\Set$).
An expression is \emph{non-Harrop} if it is not Harrop; 
it is \emph{non-computational (nc)} if it contains neither
disjunctions, nor restrictions nor
concurrency
nor free predicate variables. Every nc-formula is Harrop but not conversely.
Finally, we define, recursively, when a formula is \emph{strict}:
Harrop formulas and disjunctions are strict. A non-Harrop conjunction is strict
if either both conjuncts are non-Harrop or 
it is a conjunction of a Harrop formula and a strict formula.
%
A non-Harrop implication is strict if the premise is non-Harrop.
Formulas of the form $\diamond x\,A$ ($\diamond\in\{\forall,\exists\}$) or
$\munu(\lambda X\lambda\vec x\,A)$ ($\munu\in\{\mu,\nu\}$) are
 \emph{strict} if $A$ is strict.
Formulas of other forms (e.g., $\rt{A}{B}$,  $\Set(A)$, $X(\vec{t})$) 
are not strict.
%
The significance of these definitions is that Harropness ensures that
(a proof of) the formula will have no computational content. Strictness ensures,
among other things, that $\bot$ is not a realizer (see Sect.~\ref{sec-pe}). 

As an additional requirement for formulas to be wellformed we demand
that in formulas of the form $\rt{A}{B}$ or $\Set(B)$, $B$ must be strict.




%

%
\emph{Notation}: $P(\vec t)$ will also be written $\vec t \in P$, and
if $\Phi$ is 
$\lambda X\,Q$, then
$\Phi(P)$ stands for $Q[P/X]$.
Definitions (on the meta level) 
of the form 
$P \eqdef \munu(\Phi)$ 
($\munu\in\{\mu,\nu\}$) 
where $\Phi = \lambda X\,\lambda \vec x\,A$,
will usually be written $P(\vec x) \eqmunu A[P/X]$. 
We write $P \subseteq Q$ for 
$\forall \vec{x}\  (P(\vec{x}) \to Q(\vec{x}))$,
  $\forall x \in P\ A$ for $\forall x\  (P(x) \to  A)$, and
  $\exists x \in P\ A$ for $\exists  x\  (P(x) \land  A)$.
$\neg A \eqdef A \to \False$ where $\False\eqdef\mu(\lambda X\, X)$. 
%

\subsection{Proof rules}
\label{sec-ProofRules}
The proof rules of CFP contain those of $\IFP$,
which are the 
usual 
natural deduction rules for 
intuitionistic first-order logic with equality (see e.g. [53]), 
plus the following rules for induction and 
coinduction, where $\Phi$ is a s.p.\ operator:
%
\[
  \infer[\CL(\Phi)]
{
  \Phi(\mu(\Phi)) \subseteq \mu(\Phi)   
  }
{
 }
\qquad
  \infer[\IND(\Phi,P)]
{
  \mu(\Phi) \subseteq P 
  }
{
 \Phi(P) \subseteq P
 }
\]
\[
  \infer[\COCL(\Phi) ]
{
  \nu(\Phi) \subseteq \Phi(\nu(\Phi))   
  }
{
 }
\qquad
  \infer[\COIND(\Phi,P)]
{
  P \subseteq \nu(\Phi) 
  }
{
  P \subseteq \Phi(P)
 }
\]
%
%
The rules for restriction and 
concurrency 
are (with the earlier mentioned condition that in formulas of the form
$\rt{A}{B}$ or $\Set(B)$, $B$ must be strict): 

\medbreak
\noindent
\fbox{\small
\begin{minipage}{\textwidth}
\[
\infer[\hbox{\begin{tabular}{l}Rest-intro\\
($A, B_0, B_1$ Harrop)\end{tabular}
}]{
        \rt{A}{(B_0 \vee B_1)}
}{
A \to (B_0 \vee B_1) \ \ \     \neg A \to B_0 \wedge B_1
}
\]
\[
\begin{array}{ll}
\infer[\hbox{Rest-bind}]{
      \rt{A}{B'}
}{
 \rt{A}{B}\ \ \          B \to (\rt{A}{B'})
}
\ \ \ \ \ \ \ \ & 
\infer[\hbox{Rest-return \ \ \ }]{  
 \rt{A}{B}
}{
  B
}  \\\\
  \infer[\hbox{Rest-antimon}]{
    \rt{A'}{B}
    }{
      A' \to A \ \ \ \rt{A}{B}  
}&
  \infer[\hbox{Rest-mp}]{
    B
}{
\rt{A}{B} \ \ \    A
}
\end{array}
\]
\[
\begin{array}{ll}
  \infer[\hbox{Rest-efq}]{
  \rt{\False}{B}
}{
}
\ \ \ \ \ \ \ \ &
\infer[\hbox{Rest-stab}]{
    \rt{\neg\neg A}{B}
    }{
    \rt{A}{B}
}

\end{array}
\]

\[
  \infer[\hbox{Conc-lem}]{
  \Set(B)
}{
\rt{A}{B}   \ \ \ \     \rt{\neg A}{B}
}
\qquad
  \infer[\hbox{Conc-return}]{\ 
  \Set(A)
}{
A
}
\]
\[
  \infer[\hbox{Conc-mp}]{\
\Set(B)
}{
  A\to B\ \ \  \Set(A) 
}\ \ \ \ \ 
%
\]
\end{minipage}
}
\medbreak

In Sect.~\ref{sec-pe} we will prove that each of these rules is 
realized by a program
from our programming language in Sect.~\ref{sec-ang}. 

\subsection{Tarskian semantics, axioms and classical logic}
\label{sub-tarski}
Although we are mainly interested in the realizability 
interpretation of $\CFP$, it is important that
all proof rules of $\CFP$ are also valid w.r.t.\ a standard Tarskian 
semantics, provided we identify $\rt{A}{B}$ with $A \to B$ and 
$\Set(B)$ with $B$.


Like $\IFP$, $\CFP$ is parametric in a set $\ax$ 
of \emph{axioms}, which have to be closed nc-formulas.
%
The significance of the restriction to nc-formulas is that these
are identical to their (formalized) realizability interpretation 
(see Sect.~\ref{sec-pe}),
in particular, Tarskian and realizability semantics coincide for them.
Axioms should be chosen such that they are true in an intended 
Tarskian model.
Since Tarskian semantics admits classical logic, this means that
a fair amount of classical logic is available through axioms.
For example, for each closed nc-formula $A(\vec x)$, stability, 
$\forall \vec x\,(\neg\neg A(\vec x) \to A(\vec x))$
can be postulated as axiom.
In addition, the rule \HT{(Conc-lem)} 
is a variant of the 
classical law of excluded middle and (Rest-stab) permits stability for
arbitrary right arguments of restriction.

In our examples and case studies we will use
\HT{an instance of $\CFP$ with a sort for real numbers and}
some standard axiomatization of real closed fields formulated as a set of 
nc-formulas. 
In particular, we will freely use 
constants, operations and relations such as 
$0,1,+,-,*, <, |\cdot|, /$ and assume their expected properties as axioms 
(expressed as nc-formulas).



%% file: sec-pe.tex
We define a realizability interpretation of $\CFP$ that will enable us
to extract concurrent programs from proofs. 
Since the interpretation extends 
the one in $\IFP$~\cite{IFP}, it suffices
to define realizability for
the restriction and
concurrency operators
and prove that their 
proof
rules are realizable (Sects.~\ref{sub-partial}). 
All definitions and proofs of this section can be carried out in a formal 
system $\RCFP$ (realizability logic for $\CFP$) which is $\CFP$ 
without $\rt{}{}$ and $\Set$ but with classical logic and an 
extended first-order language 
\UB{that contains the earlier introduced programs and types as terms and
 the typing relation `$:$' as a predicate constant, 
and describes their semantics through suitable axioms.}
In particular, $\RCFP$ 
\UB{proves}
the correctness of extracted programs
(Soundness Theorem~\ref{thm-soundnessI}). Since it only matters that $\RCFP$ is
classically correct (since no realizability interpretation 
is applied to it), 
details of $\RCFP$ do not matter and are therefore omitted.

\subsection{Realizability}
\label{sub-realizability}
%
%
Realizability for $\CFP$ is formalized in $\RCFP$ and follows the
pattern in \cite{IFP}.
For every non-Harrop $\CFP$-formula $A$ a type $\tau(A)$ and a 
$\RCFP$-predicate $\rea(A)$ 
are defined such that $\rea(A)$ is a subset of $\tau(A)$
(more precisely, $\RCFP$ proves 
\HT{$\forall a (\rea(A)(a) \to a:\tau(A))$}
hence the interpretation of $\rea(A)$ is a subset of $\tval{\tau(A)}{}$).
We often write $\ire{a}{A}$ for $\rea(A)(a)$ (`$a$ realizes $A$') and 
$\re\,A$ for $\exists a\,\rea(A)(a)$ (`$A$ is realizable').

Since Harrop formulas (see Sect.~\ref{sub-cfpsyntax})
have trivial computational content, it only matters whether they are
realizable or not.
Therefore, we define for a Harrop formula $A$, a $\RCFP$-formula
 $\reah(A)$ that 
represents
the realizability interpretation of $A$, but with suppressed realizer.
Formally, we define by simultaneous recursion, for every 
Harrop $\CFP$-expression $E$ an $\RCFP$-expressions $\reah(E)$,
and for every non-Harrop $\CFP$-expressions $E$
an $\RCFP$-expressions $\rea(E)$.
%
%
It is convenient to set, in addition, for Harrop formulas
$\tau(A) \eqdef \one$ and 
$\rea(A) \eqdef \lambda a\,(a = \Nil \land \reah(A))$,
so that $\tau(A)$ and $\rea(A)$ are defined for \emph{all} $\CFP$-formulas.

\begin{figure}
\fbox{\small
\begin{minipage}{\textwidth}
For Harrop formulas $A$: 
$\tau(A) = \one$ and  $\rea(A) = \lambda a\,(a = \Nil \land \reah(A))$.
\bigskip

$\tau(E)$ for non-Harrop expressions $E$:
\begin{align*}
\tau(P(\vec t)) &= \tau(P) \qquad\quad
\tau(A \lor B) = \tau(A) + \tau(B)\\
\tau(A \land B) &= \left\{ \begin{array}{ll}
            \tau(A) \times \tau(B) &\hbox{($A,B$ non-Harrop)}\\
           \tau(A)  &\hbox{($B$ Harrop)}\\             
           \tau(B)  &\hbox{($A$ Harrop)}
                            \end{array}\right. \\
\tau(A \to B) &= \left\{ \begin{array}{ll}
      \ftyp{\tau(A)}{\tau(B)}  &\hbox{($A$ non-Harrop)}\\
      \tau(B)  &\hbox{($A$ Harrop)}
                         \end{array}\right.\\
\tau(\rt{A}{B}) &= \tau(B) \qquad 
\tau(\Set(B)) = \Am(\tau(B)) 
\\
\tau(\diamond x\,A) &= 
  \tau(A) \qquad \hbox{($\diamond \in\{\forall,\exists\}$)}\\[0.5em]
\tau(X) &= \alpha_X  \qquad\qquad
%
\HT{\tau(P) = \one \hbox{\qquad ($P$ a predicate constant)}}\\
\tau(\lambda \vec x\,A) &= \tau(A) \qquad \qquad
\tau(\munu (\lambda X\,P)) = \tfix{\alpha_X}{\tau(P)}
                  \qquad \hbox{($\munu \in\{\mu,\nu\}$)}\\
\end{align*}

$\rea(E)$ for non-Harrop expressions $E$:
\begin{align*}
\rea(P(\vec t)) &= \lambda a\,(\rea(P)(\vec t,a)) &\\ 
\rea(A\lor B)   &=\lambda c\,(\ex{a}(c=\inl{a}\land\ire{a}{A})\lor
                              \ex{b}(c=\inr{b}\land\ire{b}{B}))\\
\rea(A\land B)  &=\left\{ \begin{array}{ll}
   \lambda c\,(\exists a,b\,(c = \Pair(a,b) \land \ire{a}{A}\land \ire{b}{B}))
                            &\hbox{($A,B$ non-Harrop)}\\
              \lambda a\,(\ire{a}{A} \land \reah(B)) 
                            &\hbox{($B$ Harrop)}\\
              \lambda b\,(\reah(A) \land \ire{b}{B})
                            &\hbox{($A$ Harrop)}
                          \end{array} \right.\\ 
\rea(A\to B)    &= \left\{ \begin{array}{ll}
  \lambda c\,(c:\ftyp{\tau(A)}{\tau(B)} \land  
          \all{a}(\ire{a}{A}\to\ire{(c\,a)}{B})) 
                     &\hbox{($A$ non-Harrop)}\\ 
          \lambda b\,(b:\tau(B) \land (\reah(A) \to \ire{b}{B}))  
                     &\hbox{($A$ Harrop)}
                           \end{array}\right.\\
\rea(\rt{A}{B}) &= \lambda b\,( b\! :\! \tau(B) \land  
                                (\re\, A \to \defined{b}) \land
                               (\defined{b} \to b\,\re\,B)) 
\\
\rea(\Set(B)) &= \lambda c\, \ex{a,b}\, 
      (c = \Amb(a, b) \land a,b:\tau(B) \land (\defined{a} \lor \defined{b})\ \land \\
              &\hspace{5em} (\defined{a} \to a\, \re\, B) \land 
                            (\defined{b} \to b\, \re\, B))
\\
\rea(\allex x\,A)  &=\lambda a\,(\allex x\,(\ire{a}{A})) 
  \qquad \hbox{($\allex\in\{\forall,\exists\}$)}\\[0.5em]
%
\rea(X) &= \reali{X} \qquad\qquad
\rea(\lambda \vec x\,A) = \lambda (\vec x,a)\,(\ire{a}{A})  
\\
\rea(\munu(\lambda X\,P)) &= \munu(\lambda\reali{X}\,\rea(P))
 \qquad \hbox{($\munu\in\{\mu,\nu\}$)}\\
%
\end{align*}

$\reah(E)$ for Harrop expressions $E$:
\begin{align*}
\reah(P(\vec t)) &= \reah(P)(\vec t)\qquad
%
\reah(A\land B)  =
      \reah(A)\land \reah(B) \\
\reah(A\to B)    &= \left\{ \begin{array}{ll}
    \re\,A\to\reah(B) &\hbox{($A$ non-Harrop)}\\
    \reah(A)\to\reah(B) &\hbox{($A$ Harrop)}
                           \end{array}\right.\\
\reah(\allex x\,A)  &=\allex x\,\reah(A)
  \quad \hbox{($\allex\in\{\forall,\exists\}$)}\\[0.5em]
\reah(P) &= P\quad \hbox{($P$ a predicate constant)}\quad\qquad
\reah(\lambda \vec x\,A) = \lambda \vec x\,\reah(A) 
\\
\reah(\munu(\lambda X\,P)) &= \munu(\lambda X\,\reah_X(P))
  \qquad \hbox{($\munu\in\{\mu,\nu\}$)}
%
%
\end{align*}
\end{minipage}
}
\caption{Realizability interpretation of $\CFP$}
\label{fig-realizability}
\end{figure}

The complete definition, which is shown in Fig.~\ref{fig-realizability},
assumes that to each $\CFP$ predicate variable $X$ there are 
assigned a fresh type variable $\alpha_X$ and a fresh $\RCFP$ predicate 
variable $\reali{X}$ with one extra argument for domain elements.
Furthermore, to define realizability for the fixed points of a Harrop operator 
$\lambda X\,P$, we use the notation
\[\reah_X(P) \eqdef\reah(P[\pcv{X}/X])[X/\pcv{X}]\]
where $\pcv{X}$ is a fresh predicate constant assigned to the (non-Harrop) 
predicate variable $X$. 
This is motivated by the fact that $\lambda X\,P$ is Harrop
iff $P[\pcv{X}/X]$ is. The idea is that $\reah_X(P)$ is the same as 
$\reah(P)$ but considering $X$ as a (Harrop) predicate constant.

To see that the definitions make sense, note that a formula 
$P(\vec t)$ is Harrop iff $P$ is,
predicate variables and disjunctions are always non-Harrop, 
a conjunction is Harrop iff both conjuncts are,
an implication $A\to B$ is Harrop iff $B$ is,
and $\forall x\,A$, $\exists x\, A$, $\lambda\vec x\,A$ are Harrop 
iff $A$ is.
%
The rationale and correctness of realizability for restriction and
concurrency
are discussed in Sect.~\ref{sub-partial}.


If a formula $A$ is nc, then it is Harrop (see Sect.~\ref{sub-cfpsyntax} for definitions)
but in addition $A$ and $\reah(A)$ are syntactically identical.
In contrast, in general, a Harrop formula $A$ neither implies nor is implied by 
$\reah(A)$.
%
\begin{lemma}
\label{lem-strict}
For every $\CFP$-formula $A$:
\begin{itemize}
\item[(1)] $\tau(A)$ is a regular type.
\item[(2)] If $A$ is strict, then $\bot$ does not realize $A$, provably in $\RCFP$.
\item[(3)]  $\Amb(\bot, \bot)$ is not a realizer of $A$.
\item[(4)]  \UB{For a program $M$ that realizes $A$, t.f.a.e.:
(i) $M$ has some productive computation; 
(ii) all computations of $M$ are productive; 
(iii) $\val{M} \neq \bot$.}
\end{itemize}
%
\end{lemma}
\begin{proof}
(1) and (2) 
are easily proved by structural induction on formulas.
\UB{(3) follows from the fact that if $A$ is of the form $\Amb(B)$, then $B$
must be strict.} 
(4) is proved by (3) and Corollary~\ref{cor:ddatabot}~(\UB{3}). 
\end{proof}


\paragraph*{Remarks and examples.} 
%
The main difference of our interpretation to the 
usual realizability interpretation of intuitionistic number theory lies in the
interpretation of quantifiers. While in number theory variables range over
natural numbers, which have concrete computationally meaningful representations,
we make no general assumption of this kind,
since it is our goal to extract programs from proofs in abstract mathematics.
This is the reason why we interpret quantifiers \emph{uniformly}, that is, 
a realizer of a universal statement must be independent
of the quantified variable and a realizer of an existential statement does not contain a
witness.
A similar uniform interpretation of quantifiers can be found in the
Minlog system.
The usual definition of realizability of quantifiers in intuitionistic number theory 
can be recovered by relativization to an inductively defined predicate $\NN$
describing natural numbers in unary representation:
$$\NN(x) \eqmu  x = 0 \lor \NN(x-1) $$
which  is shorthand for 
$\NN \eqdef \mu(\lambda X\, \lambda x\, (x = 0 \lor X(x-1)))$.
The type $\tau(\NN)$ assigned to $\NN$ is the recursive type of unary natural numbers
\[\nat\eqdef  \tfix{\alpha}{1+\alpha}.\] 
Realizability for $\NN$ works out as
\[
\ire{a}{\NN(x)} \eqmu (a = \Left \land x = 0)  
\lor \exists b\,(a = \Right(b) \land \ire{b}{\NN(x-1))}\,.
\]
Thus, $\NN(0)$, $\NN(1)$, $\NN(2)$ are realized by $\Left$ (i.e., $\Left(\Nil)$),
$\Right(\Left)$, $\Right(\Right(\Left))$, and so on.
Therefore,  the (unique) realizer of $\NN(n)$ is the unary representation of $n$.
Other ways of defining natural numbers may induce different
representations.  
%
An example of a formula with interesting realizers is the formula 
expressing that the sum of two natural number is a natural number,
\begin{equation}
\label{eq:intro1}
\forall x, y\ (\NN(x) \to \NN(y) \to \NN(x+y)).
\end{equation}
It has type $\nat \to \nat \to \nat$ and
is realized by a function $f$ that, given realizers of $\NN(x)$ and $\NN(y)$, 
returns a realizer of $\NN(x+y)$, hence $f$ performs addition of unary numbers.

\begin{example}[Non-terminating realizer]
\label{ex-d}
Let
$$
\D(x) \eqdef  x\neq 0 \to (x\leq 0 \lor x\geq 0)\,.
$$  
Then 
$\tau(\D) = \bool$ where $\bool = \one + \one$, and $\ire{a}{\D(x)}$ unfolds to
$$
a: \tau(\bool) \land (x \neq 0 \to (a = \Left \land x \leq 0) \lor 
(a = \Right \land x \geq 0)).
$$
Therefore, $\D(x)$ is realized by $\Left$ if $x < 0$ and by $\Right$ if $x > 0$.
If $x=0$, any element of $\tau(\bool)$ realizes $D(x)$,
in particular $\bot$. 
Hence, nonterminating programs,
which, by Lemma \ref{lem-strict}~(4), denote $\bot$,
realize $D(x)$.
In contrast, \emph{strict} formulas are never realized 
by a nonterminating program, as shown in Lemma~\ref{lem-strict}~(2). 

\end{example}

\subsection{Partial correctness and concurrency}
\label{sub-partial}
\label{sub-conc}
We explain realizability for $\rt{A}{B}$ and $\Set(B)$ and show that the associated
proof rules are sound.

As we have seen in Example~\ref{ex-d}, a realizer of an implication
$A \to B$ where $A$ is a Harrop formula is realized by a
`conditionally correct' program $M$,
that is, if $\reah(A)$, then $M$ realizes $B$, but otherwise
no condition is imposed on $M$, in particular $M$ may be non-terminating. 
However, $M$ may terminate but fail to realize $B$. This means that
termination of a realizer of $A\to B$
is not a sufficient condition for correctness (correctness meaning to
realize $B$). But, as explained in the Introduction, 
this is what we need to concurrently realize a formula.
The definition of realizability for the new logical operator 
$\rt{}{}$ (shown in Fig.~\ref{fig-realizability})
achieves exactly this: A realizer of a restriction 
$\rt{A}{B}$ is `partially correct' in the sense that 
it is correct iff it terminates.
%
%
%
%
%
\HT{
By Lemma \ref{lem-strict}~(4),  for a program $M$ to realize
$\rt{A}{B}$ means that $M$ has type $\tau(B)$, 
\UB{and} if $A$ is realizable then all the computations of $M$ are productive,
and conversely, if 
\UB{$M$ has a productive computation}
then $M$ always (that is, independently of the realizability of $A$) realizes $B$.}

To highlight the difference between restriction and implication
in a more concrete situation, 
consider $\rt{A}{(A\lor B)}$ vs.\ $A \to (A \lor B)$
where $A$ is Harrop. Clearly $\Left$ realizes $A \to (A \lor B)$,
but in general
$\rt{A}{(A\lor B)}$ is not realizable.
Note that $\Left$ even
realizes $A \stackrel{\mathrm{u}}{\to} (A \lor B)$ where
$\stackrel{\mathrm{u}}{\to}$ is Schwichtenberg's uniform 
implication~\cite{SchwichtenbergWainer12}, 
hence restriction is also different from uniform implication.

The intuition of $\Amb(a,b)$ realizing $\Set(A)$ 
is that it is a pair of candidate realizers at least one of which 
is productive,  and each productive one is a realizer.

\begin{lemma}
\label{lem-restrict}
The rules for restriction and concurrency are realizable.
\end{lemma}
\begin{proof}
%
The table below shows the realizers of each rule for the (most interesting) 
case where the conclusion is non-Harrop, using the definitions
\begin{eqnarray*}
\leftright\ \eqdef \lambda b. \caseof{b}    \{\Left(\_) \to \Left; \Right(\_) \to \Right\}\,,\\
\mapamb\eqdef \ \lambda f.\, \lambda c.\ \caseof{c}\,
\{\Amb(a,b) \to 
 \Amb(\strictapp{f}{a}, \strictapp{f}{b})\}\,.
\end{eqnarray*}
Proofs of their correctness are
in Appendix~\ref{appendix-proofs}. 
For (Rest-intro), (Rest-stab), and (Conc-lem), classical logic is needed.
Here, we set $a\, \seq\, b \eqdef \strictapp{(\lambda c.\ b)}{a}$.
\end{proof}

\medbreak
\noindent
\fbox{
\begin{minipage}{\textwidth}
\[
\infer[\hbox{\begin{tabular}{l}Rest-intro
($A, B_0, B_1$ Harrop)\end{tabular}
}]{
        \ire{(\leftright\ b)}{\rt{A}{(B_0 \vee B_1)}}
}{
\ire{b}{(A \to (B_0 \vee B_1))} \ \ \     \reah({\neg A \to B_0 \wedge B_1})
}
\]
\[
\begin{array}{ll}
\infer[\hbox{\begin{tabular}{l}
Rest-bind\ ($B$ non-Harrop)\\
\small ($\ire{(a\,\seq\,f)}{\rt{A}{B'}}$ ($B$ Harrop))
\end{tabular}}]
{      \ire{(\strictapp{f}{a})}{\rt{A}{B'}}
}{
\ire{a}{ \rt{A}{B}}\ \ \          \ire{f}{( B \to (\rt{A}{B'}))}
}
& 
\infer[\hbox{Rest-return \ \ \ }]{  
 \ire{a}{\rt{A}{B}}
}{
  \ire{a}{B}
} \end{array}
\]
\[
\begin{array}{ll}
  \infer[\hbox{Rest-antimon}]{
    \ire{a}{\rt{A'}{B}}
    }{
      \ire{}{(A' \to A)} \ \ \ \ire{a}{\rt{A}{B}  }
}\ \ \ \ \ \ \ &
  \infer[\hbox{Rest-mp}]{
    \ire{b}{B}
}{
\ire{b}{\rt{A}{B}} \ \ \    \ire{}{A}
}
\end{array}
\]
\[
\begin{array}{ll}
  \infer[\hbox{Rest-efq}]{
  \ire{\bot}{\rt{\False}{B}}
}{
}
\ \ \ \ \ \ \ \ &
\infer[\hbox{Rest-stab}]{
    \ire{b}{\rt{\neg\neg A}{B}}
    }{
    \ire{b}{\rt{A}{B}}
}

\end{array}
\]

\[
  \infer[\hbox{Conc-lem}]{
  \ire{\Amb(a,b)}{\Set(B)}
}{
\ire{a}{\rt{A}{B}}   \ \ \ \     \ire{b}{\rt{\neg A}{B}}
}
\qquad
  \infer[\hbox{Conc-return}]{\ 
  \ire{\Amb(a,\bot)}{\Set(A)}
}{
\ire{a}{A}
}
\]
\[
  \infer[\hbox{\begin{tabular}{l}
Conc-mp\ ($A$ non-Harrop)\\
\small ($\ire{\Amb(f, \bot)}{\Set(B)}$ ($A$ Harrop))
\end{tabular}}]
{\
\ire{(\mapamb\ f\ c)}{\Set(B)}
}{
  \ire{f}{(A\to B)}\ \ \  \ire{c}{\Set(A) }
}\ \ \ \ \ 
\]
\smallskip
\end{minipage}
}

\begin{lemma}\label{class-orelim-sleep}
$\CFP$ derives the following rules. 
The rules are displayed together with their extracted realizers.
%
\begin{itemize}
 \item[(1)]
\raisebox{-0.2cm}{
$  \infer[]{
   \ire{\Amb(\strictapp{\Left}{a},\strictapp{\Right}{b})}{\Set(B_0\lor B_1)}
 }{
    \ire{a}{\rt{A_0}{B_0}} \ \ \ \  \ire{b}{\rt{A_1}{B_1}}   \ \ \ \  \reah{(\neg\neg(A_0\lor A_1))}
 }$
}
\item[(2)]
\raisebox{-0.2cm}{
$
  \infer[]{
   \ire{\caseof{a} \{\Left(\_) \to \bot;\Right(b) \to b\}}{\rt{D\land\neg B}{C}}
 }{
    \ire{a}{\rt{D}{(B \lor C)}}
 }$}
\qquad {($C$ strict)}

\end{itemize}
\end{lemma}

\begin{example}
\label{ex-dprime}
Continuing Example~\ref{ex-d}, we modify $\D(x)$ to 
$$\D'(x) \eqdef  \rt{x\neq 0}{(x\leq 0 \lor x\geq 0)}\,.$$
A realizer of $\D'(x)$, which has type $\bool$, may or may not terminate
(non-termination occurs when $x = 0$).
However, in case of termination,  the result is guaranteed to 
realize $x\le 0 \lor x\ge 0$.
Note that, a realizer of $\D(x)$ also has type $\bool$ and may or may not terminate,
but there is no guarantee that
it realizes $x\le 0 \lor x\ge 0$ when it does terminate.
Nevertheless, 
$\D \subseteq \D'$ 
follows from (Rest-intro) 
(since $\neg x\ne 0$ 
implies
$x\le 0 \land x\ge 0$) 
and is realized by $\leftright$.
%
$\D' \subseteq \D$ holds trivially.
\end{example}
\begin{example}
\label{example-ConsSD}
This 
builds on the examples~\ref{ex-d} and~\ref{ex-dprime}
and will be used in 
Sect.~\ref{sec-gray}.
Let $\tent(x) = 1-2|x|$ and consider the predicates
$\E(x) \eqdef  \D(x) \land \D(\tent(x))$ and
%
\[
\ConSD(x)  \eqdef  \Set((x\leq 0 \lor x\geq 0) \lor |x| \leq 1/2).
\]
%
We show $\E \subseteq \ConSD$:
From $\E(x)$ and Example~\ref{ex-dprime} we get $\D'(x)$ and $\D'(\tent(x))$
which unfolds to $\rt{x\neq 0}{(x\le 0  \lor x \ge 0 )}$ and 
$\rt{|x|\neq 1/2}{(|x| \geq 1/2  \lor |x| \leq 1/2 )}$. 
By Lemma \ref{class-orelim-sleep}~(2),
$\rt{|x| < 1/2}{(|x| \leq 1/2)}$. 
%
Since $\neg\neg((x \ne 0) \lor |x| < 1/2)$, we have $\ConSD(x)$ 
by  
Lemma \ref{class-orelim-sleep}~(1).
 Moreover, $\tau(\E) = \bool\times \bool$ and 
$\tau(\ConSD) = \Am(\tri)$  where $\tri \eqdef \bool + \one$.
The extracted realizer of $\E \subseteq \ConSD$ is
\begin{align*}
&\conSD \eqdef \lambda c. \caseof{c} \{\Pair(a, b) \to  
   \Amb(\strictapp{\Left}{(\leftright\ a)},\\
  & \hspace*{4.2cm} 
  \strictapp{\Right}{(\caseof{b}\{\Left(\_) \to \bot; \Right(\_) \to \Nil\})})\}   
\end{align*}
of type $\tau(\E \subseteq \ConSD) =  \bool\times \bool \to \Am(\tri)$.
Explanation of this program:
$a$ is $\Left$ or $\Right$ depending on whether $x \leq 0$ or $x \geq 0$
but may also be $\bot$ if $x = 0$.
$b$ is $\Left$ or $\Right$ depending on whether $|x| \leq 1/2$ or $|x| \geq 1/2$ 
but may also be $\bot$ if $|x| = 1/2$.  Since $x = 0$ and $x = 1/2$ do not
happen simultaneously, by evaluating $a$ and $b$ concurrently, we obtain one of them
from which we can determine one of the cases $x \leq 0$, $x \geq 0$, or $|x| \leq 1/2$.
\end{example}

\subsection{Soundness and program extraction}
\label{subsec:soundness}
%

As we did in the above example, one can extract 
from any CFP-proof of a formula a program that realizes it.
This property is called the Soundness Theorem of realizability. Its
proof is the same as for IFP~\cite{IFP} but extended by 
the rules for the new logical operators whose realizability we proved in 
Sects.~\ref{sub-partial}.
%
\begin{theorem}[Soundness Theorem I]
\label{thm-soundnessI}
%
From a 
$\CFP$-proof 
of a formula $A$ from a set of axioms
one can extract a 
program $M$ 
of type $\tau(A)$ (which is a regular type)
such that $\RIFP$ proves $\ire{M}{A}$ 
from the same axioms. 
%
\end{theorem}
%
%

%
In $\CFP$, we have a second Soundness Theorem 
which ensures the correctness
of all results of fair computation paths of an extracted program $M$. 
%
More precisely, correctness of $M$ means that all $d\in\ddata(\val{M})$ 
realize the formula $A^-$ obtained from $A$ by deleting all concurrency 
operators $\Set$.
Since $A^-$ is an $\IFP$ formula, the
Theorem relates the realizability interpretations of $\CFP$ and $\IFP$.



\UB{
However, such a correctness result only holds for formulas whose realizers do not
contain $\Amb$ in the scope of a lambda-abstraction.
This restriction is enforced by the following syntactic 
admissibility condition: 
}
\HT{
An expression is called \emph{admissible} if it contains neither
free predicate variables nor restrictions ($\rt{}{}$), and 
all occurrences of concurrency ($\Set$) are strictly positive and
at non-F-position.
Here, \UB{the notion of a} \emph{subexpression at F-position} 
in a $\CFP$ expression is defined inductively by three rules:}
(i) A subexpression of the form $A\to B$ where $A$ and $B$ are both non-Harrop
is at F-position.
\mps{\HT{I changed the order of (ii) and (iii)}}
\HT{
(ii) A subexpression $\munu\,\lambda X\,Q$ ($\munu\in\{\mu,\nu\}$) is at 
F-position if $Q$ has a free occurrence of $X$ at F-position.
(iii) A subexpression within a subexpression at F-Position is at F-position.
}

For example, 
$\Set(\mu(\lambda X\,\lambda x\,(x=0 \lor \forall y\,(\NN(y)\to X(f(x,y))))))$
is admissible, whereas
$\mu(\lambda X\,\lambda x\,\Set(x=0 \lor \forall y\,(\NN(y)\to X(f(x,y)))))$
is not. The predicate $\ConSD$ in Example~\ref{example-ConsSD} is admissible.
%
\begin{theorem}[Faithfulness]
\label{thm-faithfulness}
If $a\in D$ realizes an admissible formula $A$, 
then all $d\in\ddata(a)$ realize $A^-$.
\end{theorem}
Theorems~\ref{thm-soundnessI} and~\ref{thm-faithfulness} imply:
 \begin{theorem}[Soundness Theorem II]
 \label{thm-soundnessII}
 From a $\CFP$ proof of an admissible 
formula $A$ from a set of axioms one can extract a 
program $M :\tau(A)$ such that 
$\RCFP$ proves 
$\forall d\in\ddata(\val{M})\,\ire{d}{A^-}$ from the same set of axioms.
 \end{theorem}
Thms.~\ref{thm-soundnessII} and~\ref{thm:data}, together with 
and classical soundness (see Sect.~\ref{sub-tarski}), yield:
 \begin{theorem}[Program Extraction]
 \label{thm-pe}
 From a $\CFP$ proof of an admissible 
formula $A$ from a set of axioms
one can extract a 
program $M:\tau(A)$ such that for any 
computation 
$M =  M_0 \newprintp M_1 \newprintp  \ldots$,
$\sqcup_{i \in \NN} (M_i)_{D}$ realizes  $A^-$ in every  
model of the axioms.

\end{theorem}

%% file: sec-gray.tex
As our main case study, we extract a concurrent conversion program between
two representations of real numbers 
in [-1, 1],
the signed digit representation and infinite Gray code.
In the following, we also write $d:p$ for $\Pair(d,p)$.

The signed digit representation is an extension of the usual binary expansion
that uses the set $\SD \eqdef \{-1, 0, 1\}$ of \emph{signed digits}. 
The following predicate $\C(x)$ expresses coinductively
that $x$ has a signed digit representation.
\begin{eqnarray*}
  \C(x) &\eqnu& |x|\le 1 \land \exists\, d \in \SD\, \C(2x-d)\,,
\end{eqnarray*}
with $\SD(d) \eqdef (d = -1 \lor d = 1) \lor d = 0$.
The type of $\C$ is $\tau(\C) = \stream{\tri}$
where $\tri \eqdef (\one + \one) + \one$ and 
$\stream{\delta} \eqdef  \tfix{\alpha}{\delta \times \alpha}$,
and its realizability interpretation is
\begin{eqnarray*}
  \ire{p}{\C(x)} &\eqnu& |x|\le 1\land\exists\, d\in\SD\,\exists p'\ (
  p = d: p'\ \land 
  \ire{p'}{\C(2x-d)})\, 
\end{eqnarray*}
which expresses indeed that $p$ is a signed digit representation of $x$, that is,
$p = d_0:d_1:\ldots$ with
$d_i\in\SD$ and
$x = \sum_{i}d_i2^{-(i+1)}$.
%
Here, we identified the three digits $d = -1, 1, 0$ with their realizers
$\Left(\Left),    \Left(\Right), \Right$.

Infinite Gray code (\cite{Gianantonio99,Tsuiki02}) is an almost redundancy free
representation of real numbers in [-1, 1] using the partial digits 
$\{-1,1, \bot\}$. A stream 
$p =d_0:d_1:\ldots$ 
of such digits
is an infinite Gray code of $x$ iff 
$d_i = \sgb(\tent^{i}(x))$
 where $\tent$ is the tent function 
$\tent(x) = 1-|2x|$ and $\sgb$ is 
a multi-valued version of the sign function for which 
$\sgb(0)$ is any element of $\{-1, 1,\bot\}$
  (see also Example~\ref{example-ConsSD}).
One easily sees that
$\tent^i(x) = 0$ for at most one $i$. 
Therefore, 
this coding has 
little 
redundancy in that the
code is uniquely determined and total except for at most one digit
which may be undefined. 
Hence, infinite Gray code 
is accessible through concurrent computation with two threads.
%
The coinductive predicate 
\begin{eqnarray*}
\G(x) &\eqnu& |x|\le 1 \land \D(x) \land \G(\tent(x))\,,
\end{eqnarray*}
where $\D$ is the predicate 
$
\D(x) \eqdef  x\neq 0 \to (x\leq 0 \lor x\geq 0)\,
$  
from Example~\ref{ex-d},
expresses that $x$ has an infinite Gray code 
(identifying $-1, 1, \bot$ with $\Left, \Right, \bot$).
Indeed,
$\tau(\G) = \bool^{\omega}$ and
\[
\ire{p}{\G(x)} \eqnu |x|\le 1\land \,\exists d,p'(p = d:p' \land
  (x\neq 0 \to \ire{d}{(x \le 0 \lor x \ge 0)})\ \land 
  \ire{p'}{\G(\tent(x))})\,.
\]

In \cite{IFP}, 
the inclusion $\C \subseteq \G$ was proved in IFP and a sequential conversion
function from signed digit representation to infinite Gray code extracted.
On the other hand, a program producing a signed digit representation from an
infinite Gray code cannot access its input sequentially 
from left to right since it will diverge when it accesses $\bot$.
Therefore, the program needs to evaluate 
two consecutive digits concurrently to obtain at least one
of them. 
With this idea in mind, we define
a concurrent version of $\C$ as
\begin{eqnarray*}
\C_2(x) &\eqnu& |x|\le 1 \land \Set(\exists\, d \in \SD\, \C_2(2x-d))\,
\end{eqnarray*}
with $\tau(\C_2) = \tfix{\alpha}{\Am(\tri \times \alpha)}$ and
prove $\G \subseteq \C_2$ in CFP (Thm.~\ref{thm-g-mc}).
Then we can extract from the proof a concurrent algorithm that 
converts infinite Gray code to signed digit representation.  
Note that, while the formula $\G \subseteq \C_2$ is \emph{not} admissible 
(it contains $\Set$ at an F-position), the formula $\C_2(x)$ \emph{is}.
Therefore, if for some real number $x$ we can prove $\G(x)$,
the proof of $\G \subseteq \C_2$ will give us a proof of $\C_2(x)$ to which
Theorem~\ref{thm-pe} applies. Since $\C_2(x)^-$ is $\C(x)$, this means that 
we have a nondeterministic program all whose fair computation paths will result in a
(deterministic) signed digit representation of $x$.

%
%


Now we carry out the proof of $\G \subseteq \C_2$.
For simplicity, we use 
pattern matching on constructor expressions for defining functions.  
For example, we write
${\mathsf f}\ (a:t) \eqdef M$
for
${\mathsf f} \eqdef \lambda x.\, \caseof{x}\,\{\Pair(a, t) \to M \}.$


The crucial step in the proof is accomplished by
Example~\UB{\ref{example-ConsSD}}, since it yields nondeterministic
information about the first digit of the signed digit representation
of $x$, as expressed by the predicate
\[
\ConSD(x)  \eqdef  \Set((x\leq 0 \lor x\geq 0) \lor |x| \leq 1/2).
\]

\begin{lemma}
\label{lem-gsd}
%
$\G \subseteq \ConSD$.
\end{lemma}
\begin{proof}
$\G(x)$ implies $\D(x)$ and $\D(\tent(x))$, 
and hence $\ConSD$, by Example~\ref{example-ConsSD}.
\end{proof}
The extracted program $\gscomp:\ftyp{\stream{\bool}}{\Am(\tri)}$  
uses the program $\conSD$ defined in Example~\ref{example-ConsSD}:
$$
\gscomp\ (a:b:p)  \eqdef \conSD\  (\Pair(a, b))\,.
$$

We also need the following closure properties of $\G$:
\begin{lemma}
\label{lem-gclosure}
Assume $\G(x)$. Then:
\begin{itemize}
\item[(1)] $\G(\tent(x))$, $\G(|x|)$, and $\G(-x)$;
\item[(2)] if $x \ge 0$, then $\G(2x-1)$ and $\G(1-x)$;
\item[(3)] if $|x|\le 1/2$, then $\G(2x)$.
\end{itemize}
\end{lemma}
\begin{proof}
This follows directly from the definition of $\G$ and 
elementary properties
of the tent function $\tent$.  
The extracted programs consist of simple manipulations of 
the given digit stream realizing $\G(x)$, concerning only its tail and  
first two digits. No nondeterminism is involved.
A detailed proof is in 
Appendix~\ref{appendix-proofs}.
\end{proof}

\begin{theorem}
\label{thm-g-mc}
$\G\subseteq\C_2$.
\end{theorem}
\begin{proof}
By coinduction. Setting 
$A(x) \eqdef \exists d \in \SD\, \G(2x-d)$, 
we have to show 
\begin{equation}
\label{eq-step}
\G(x) \to |x| \leq 1 \land \Set(A(x))\,.
\end{equation}
Assume $\G(x)$. 
Then $\ConSD(x)$, by Lemma \ref{lem-gsd}.
Therefore, it suffices to show 
\begin{equation}
\label{eq-consd}
\ConSD(x)  \to \Set(A(x))\,
\end{equation} 
which, with the help of the rule (Conc-mp), can be reduced to
\begin{equation}
\label{eq-onedigit}
(x \leq 0 \lor x \geq 0 \lor |x| \leq 1/2)   \to A(x).
\end{equation}
(\ref{eq-onedigit}) can be easily shown 
using Lemma~\ref{lem-gclosure}:
If $x\le 0$, then $\tent(x) = 2x+1$. Since $\G(\tent(x))$, we have 
$\G(2x-d)$ for $d=-1$.
If $x \ge 0$, then $\G(2x-d)$ for $d=1$ by (2).
If $|x|\le 1/2$, then $\G(2x-d)$ for $d=0$ by (3).
\end{proof}


The program 
$\onedigit : \ftyp{\stream{\bool}}{\ftyp{\tri}{\tri\times\stream{\bool}}}$
extracted from the proof of (\ref{eq-onedigit}) 
from the assumption $\G(x)$ is
\begin{align*}
\onedigit\ (a:b:p)\ c \eqdef 
&  \caseof{c} \{\Left(d) \to 
                       \caseof{d} \{\\
&\hspace{3.5cm} \Left(\_) \to \Pair(-1, b:p); \\
&\hspace{3.5cm}               \Right(\_) \to \Pair(1,(\nnot\ b):p) \};\\
&\hspace{1.5cm}\Right(\_) \to    \Pair(0, a: (\nh\ p))\}\\
\nnot\ a \eqdef& \caseof{a} \{\Left(\_) \to \Right; \\
&\hspace*{1.6cm}\Right(\_) \to \Left\}\\
\nh\ (a: p) \eqdef& (\nnot\ a): p
\end{align*}
This is lifted to a proof of (\ref{eq-consd}) using $\mapamb$ (the realizer of
(Conc-mp)). Hence the extracted realizer 
${\mathsf s} : \ftyp{\stream{\bool}}{\Am(\tri \times \stream{\bool})}$
of~(\ref{eq-step}) is 
\[
{\mathsf s}\ p \eqdef  \mapamb\ (\onedigit\ p)\  (\gscomp\ p)
\]
%

The main program extracted  from the proof of  Theorem~\ref{thm-g-mc} is
obtained from the step function $\mathsf{s}$ by a special form of recursion, 
commonly known as \emph{coiteration}. Formally, we use the realizer of the
coinduction rule $\COIND(\Phi_{\C_2},\G)$ where $\Phi_{\C_2}$ is the operator 
used to define $\G$ as largest fixed point, i.e.\ 
\begin{align*}
\Phi_{\C_2} &\eqdef \lambda X\,\lambda x\, |x| \leq 1 \land \Set(\exists d \in \SD\, 
       X(2x-d)).
\end{align*}
The realizer of coinduction (whose correctness is shown in~\cite{IFP}) 
also uses a program 
$\mon : \ftyp{(\ftyp{\alpha_X}{\alpha_Y})}{\Am(\ftyp{\tri \times \alpha_X)}{\Am(\tri \times \alpha_Y)}}$ 
extracted from the canonical proof of 
the monotonicity of $\Phi_{\C_2}$:
\begin{align*}
&\mon\ f\ p = \mapamb\ (\mon'\ f)\  p\\ 
&\qquad \hbox{where}\quad
\mon'\ f\ (a:t) = a : f\ t
\end{align*}
Putting everything together, we obtain the 
\emph{infinite Gray code to signed digit representation conversion program} 
$\gtos : \ftyp{\stream{\bool}}{\tfix{\alpha}{\Am(\tri \times \alpha)}}$
\begin{align*}
&\gtos \eqrec (\mon\ \gtos) \circ \mathsf{s}
\end{align*}
%


Using the equational theory of RIFP,
one can simplify $\gtos$ to the following program.
The soundness of RIFP axioms with respect to the denotational semantics and the
adequacy property of our language guarantees that these two programs are equivalent.
%
%
%
\begin{align*}
\gtos\ &(a:b:t) =\Amb(\\
&\ \ (\caseof{a} \{\Left(\_) \to  -1: \gtos\ (b:t);\\
&\ \ \hspace*{1.7cm}                   \Right(\_) \to  1: \gtos((\nnot\ b):t)\}),\\
&\ \        (\caseof{b} \{\Right(\_) \to  0: \gtos(a:(\nh\ t))\})).\\
&\ \   \hspace*{1.7cm}      \Left(\_) \to  \bot\})).
\end{align*}

In \cite{Tsuiki05}, a Gray-code to signed digit conversion program was written
with the locally angelic $\Amb$ operator
that evaluates the first two cells $a$ and $b$  in parallel and continues the computation based
on the value obtained first.  In that program,  if the value of $b$ is first obtained and 
it is $\Left$, then it has to evaluate $a$ again.   With globally angelic choice, as the above program shows, 
one can simply neglect the value 
to use the value of the other thread.  
Globally angelic choice also has the possibility to speed up the computation  if 
the two threads of $\Amb$ are computed in parallel and 
the {\em whole} computation based on the secondly-obtained value of $\Amb$ terminates first.


%% file: sec-experiments.tex
Since our programming language can be viewed as a fragment of Haskell,  
we can execute the extracted program in Haskell by
implementing the Amb operator with the Haskell concurrency module.
We comment on the essential points of the implementation.
\HT{The full code is available from \cite{githubUB}.}

First, we define the domain $D$ as a Haskell data type: 
\begin{verbatim}
data D = Nil | Le D | Ri D | Pair(D, D) | Fun(D -> D) | Amb(D, D)
\end{verbatim}
The $\ssp$-reduction, which preserves the Phase I denotational semantics and reduces a program to a w.h.n.f. 
with the leftmost outermost reduction strategy,  coincides with reduction in Haskell. 
Thus, we can identify extracted programs with programs of type D that compute 
that phase.

The $\newprintc$ reduction that concurrently calculates the arguments of $\Amb$
can be implemented  with the Haskell concurrency module.   
In \cite{BoisPLT02}, the (locally angelic) amb operator 
was implemented in Glasgow Distributed Haskell (GDH).
Here, we implemented it 
with  the Haskell libraries \verb|Control.Concurrent| and \verb|Control.Exception|
as 
a simple function 
\verb|ambL :: [b] -> IO b|
that concurrently evaluates the elements of a list and writes the 
result first obtained in a mutable variable.

Finally, the function 
\verb|ed :: D -> IO D|  produces an element of  $\ddata(a)$  from $a \in D$
by activating \verb|ambL| for the case of $\Amb(a, b)$.
It corresponds to $\newprintp$-reduction though it computes arguments of 
a pair 
sequentially.
This function is nondeterministic since the result of executing \verb|ed (Amb a b)|
depends on which of the arguments \verb|a,b| delivers a result first.
The set of all possible results of \verb|ed a| corresponds to the set $\ddata(a)$.

We executed the program extracted in Section \ref{sec-gray} with
\verb|ed|. 
As we have noted, the number $0$ has three Gray-codes
(i.e., realizers of $\G(0)$):  
$a = \bot\!:\!1\!:\!(-1)^\omega$,
$b = 1\!:\!1\!:\!(-1)^\omega$, and 
$c = -1\!:\!1\!:\!(-1)^\omega$.  On the other hand, 
the set of signed digit representations of $0$ is
$A \cup B \cup C$ where
$A = \{0^\omega\}$,
$B = \{0^k\!:\!1\!:\!(-1)^\omega \mid k \geq 0\}$, 
and
$C = \{0^k\!:\!(-1)\!:\!1^\omega \mid k \geq 0\}$, 
i.e.,
$A \cup B \cup C$ is the set of realizers of $\C(0)$.
One can calculate
$$\gtos (a) = \Amb(\bot, 0\!:\! \Amb(\bot,  0\!:\! \ldots))$$
and $\ddata(\gtos (a)) = A$. Thus 
$\gtos (a)$ is reduced uniquely to $0 \!:\! 0 \!:\! \ldots$
by the operational semantics.
On the other hand,  one can calculate $\ddata(\gtos (b)) = A \cup B$ and
$\ddata(\gtos (c)) = A \cup C$.  They are subsets of 
the set of realizers of $\C(0)$ as Theorem \ref{thm-soundnessII} says, and
$\gtos (b)$ is reduced to an element of $A \cup B$ 
as Theorem \ref{thm-pe} says.

We wrote a program that produces a $\{-1, 1, \bot\}$-sequence 
with the speed of computation of each digit ($-1$ and $1$) be controlled.  
Then, apply it to \verb|gtos| and then to \verb|ed| to obtain expected results. 

%% file: sec-conclusion.tex
We introduced the logical system $\CFP$ by extending $\IFP$ \cite{IFP}
with two propositional operators $\rt{A}{B}$ and $\Set(A)$,
and developed a method for
extracting nondeterministic and concurrent 
programs that are provably total and  satisfy
their specifications. 


While $\IFP$ already imports classical logic through nc-axioms that
need only be true classically, in $\CFP$ the access to classical logic
is considerably widened through the rule (Conc-lem) which,
when interpreting $\rt{A}{B}$ as $A \to B$ and identifying $\Set(A)$ with $A$,
is constructively 
invalid but has nontrivial 
\HT{nondeterministic}
computational content.

We applied our system to extract a 
\HT{concurrent}
translation from infinite Gray 
code to the signed digit representation, thus demonstrating that this approach
not only is about program extraction `in principle' but can be used to
solve nontrivial 
\HT{concurrent}
computation problems through 
program extraction.
%

After an overview of related work, 
we conclude with some 
ideas
for follow-up research.

\subsection{Related work}
\label{sub-related}
The CSL 2016 paper \cite{BergerCSL16} 
is an early attempt to capture concurrency via program extraction 
and can be seen as the starting point of our work. Our main advances, compared 
to that paper, are that
it is formalized as a logic for concurrent execution of partial programs
by a globally angelic choice operator which is formalized by introducing a new connective $B|_A$, and
 that we are able to express bounded nondeterminism with complete
 control of the number of threads while
\cite{BergerCSL16} 
 modelled nondeterminism with
countably infinite branching, which is unsuitable 
or an overkill for most applications. Furthermore, our approach has a typing discipline,
a sound and complete small-step reduction,  and has 
 the ability to switch 
between global and local nondeterminism (see Sect.~\ref{sub-local} below).


As for the study of angelic nondeterminism,  it is not easy to develop a denotational 
semantics as we noted in Section \ref{sec-ang}, and
it has been mainly studied from the operational point of view,
e.g.,  notions of equivalence or refinement of processes and associated proof methods, which are all fundamental for correctness and termination
\cite{LassenMoran99,MoranSandsCarlsson2003,Lassen2006,sabel_schmidt-schauss_2008,CarayolHirschkoffSangiori2005,Levy07}.
Regarding 
imperative languages, 
Hoare logic and its extensions have been applied to nondeterminism and proving 
totality from the very beginning (\cite{Apt2019FiftyYO} is a good survey on this subject).
\cite{Mamouras15} studies angelic nondeterminism with an extension of Hoare Logic.

There are many logical approaches to concurrency.
An example is an approach based
on extensions of Reynolds' separation logic~\cite{Reynolds:2002} to  
the concurrent and higher-order setting~\cite{OHearn07,Brookes07,Jungetal18}.
Logics for session types
and process calculi ~\cite{Wadler14a,CairesPfenningToninho16,Kouzapasetal16}
form another approach
that is oriented more towards the 
formulae-as-types/proofs-as-programs~\cite{Howard80,Wadler14} or rather
proofs-as-processes paradigm~\cite{Abramsky94}.
All these approaches provide highly specialized logics and expression languages
that are able to model and reason about concurrent programs with a fine control 
of memory and access management and complex communication patterns.
%


\subsection{Modelling locally angelic choice} 
\label{sub-local}
We remarked earlier that our interpretation of $\Amb$ corresponds to
\emph{globally} angelic choice. Surprisingly, \emph{locally} angelic choice
can be modelled by a slight modification of the restriction and 
the total concurrency
operators: We simply replace $A$ by the logically equivalent formula
$A \lor \False$, more precisely, we set
$\rtp{A}{B} \eqdef \rt{A}{(B\lor\False)}$ and
$\Set'(A) \eqdef \Set(A\lor\False)$.
Then the proof rules in 
Sect.~\ref{sec-cfp}
with $\rt{}{}$ and $\Set$ replaced by $\rtp{}{}$ and
$\Set'$, respectively but without the 
strictness
condition, are theorems of $\CFP$. 
To see that the operator $\Set'$ indeed corresponds to locally angelic choice
it is best to compare the realizers of the rule (Conc-mp) for $\Set$ and $\Set'$.
Assume $A$, $B$ are non-Harrop and $f$ is a realizer of $A \to B$.
Then, 
if $\Amb(a,b)$ realizes $\Set(A)$, 
then $\Amb(\strictapp{f}{a}, \strictapp{f}{b})$ realizes $\Set(B)$.
This means that to choose, say, the left argument of $\Amb$ as a result,
$a$ must terminate and so must the ambient (global) computation 
$\strictapp{f}{a}$.
On the other hand, the program extracted from the proof of (Conc-mp) for
$\Set'$ takes a realizer $\Amb(a,b)$ of $\Set'(A)$ and returns 
$\Amb(\strictapp{(\aup \circ f \circ \adown)}{a}, \strictapp{(\aup \circ f \circ \adown)}{b})$
as realizer of $\Set'(B)$, 
where $\aup$ and $\adown$ are the realizers of $B \to (B\lor\False)$ and 
$(A \lor\False) \to A$, 
namely,
$\aup \eqdef \lambda a.\, \Left(a)$ and
$\adown \eqdef \lambda c.\, \caseof{c}\{\Left(a) \to a\}$.
%
Now, to choose the left argument of $\Amb$, 
it is enough for $a$ to terminate since the non-strict operation $\aup$
will immediately produce a w.h.n.f. 
without invoking the ambient computation.
%
%
By redefining realizers of $\rt{A}{B}$ and $\Set(A)$ as realizers of 
$\rtp{A}{B}$ and $\Set'(A)$ and 
the realizers of the rules of $\CFP$ as those extracted from the proofs 
of the corresponding rules for $\rtp{}{}$ and $\Set'$,
we have another realizability interpretation of CFP that models
locally angelic choice.  


\subsection{Markov's principle with restriction}
\label{sub-markov}
So far, (Rest-intro) is the only rule that derives a restriction in a 
non-trivial way. 
However, there are other such rules,
for example
\begin{center}
\AxiomC{$\forall x \in \NN (P(x) \lor \neg P(x))$}
\RightLabel{Rest-Markov}
\UnaryInfC{$\rt{\exists x \in\NN\,P(x)}{\exists x \in\NN\,P(x)}$}
            \DisplayProof \ \ \ \ 
          \end{center}
%
If $P(x)$ is Harrop, then (Rest-Markov) 
is realized by minimization.
More precisely, if $f $ realizes $\forall x \in \NN (P(x) \lor \neg P(x))$,
then $\min(f)$ realizes the formula  
$\rt{\exists x \in\NN\, P(x)}{\exists x \in\NN\,P(x)}$,
where $\min(f)$ computes the least $k \in \NN$ such that $f\, k = \Left$
if such $k$ exists, and does not terminate, otherwise. 
One might expect as conclusion of (Rest-Markov) the formula
$\rt{(\neg\neg\exists x \in\NN\,P(x))}{\exists x \in\NN\,P(x)}$.  
However, because of (Rest-stab) (which is realized by the identity), 
this wouldn't make a difference.
The rule (Rest-Markov) can be used, for example, to prove that
Harrop predicates that are recursively enumerable (re) and 
have re complements are decidable. 
From the proof one can extract a program 
that 
concurrently searches for evidence of membership in the predicate and
its complement.
%



\subsection{Further directions for research}
\label{sub-further}
\mps{\HT{Is it possible to refer to the Gaussian paper we have just written?}
\UB{I cited it as 'to appear'.}}
The undecidability of equality of real numbers, which is 
at the heart of
our case study on infinite Gray code, is also a critical point in 
Gaussian elimination where one needs to find a non-zero entry in a non-singular
matrix. As shown in~\cite{BergerSeisenbergerSpreenTsuiki22}, 
our approach makes it possible to search for such `pivot elements'
in a concurrent way. 
%
A further promising research direction is to extend the work on coinductive
presentations of compact sets in~\cite{Spreen20} to the concurrent setting.
%



%% file: appendix-proofs.tex
\textbf{Lemma \ref{lem:ddatabot}.}
If $a \in D$ belongs to a regular type, then the following are equivalent:
  (1) $a \in\{ \bot,\Amb(\bot, \bot)\}$;   
  (2) $\{\bot\} = \ddata(a)$;
  (3) $\bot \in \ddata(a)$.

\begin{proof}
By the definition of $\ddata(a)$, clearly (1) implies (2). Since (2) implies (3)
trivially, it only remains to show that (3) implies (1), and it is only for 
this implication that we need the assumption that $a$ belongs to a regular type.
Assume $\bot\in\ddata(a)$. Then, clearly $a=\bot$, or $a=\Amb(a',b')$.
In the first case we are done. In the second case we need to show $a'=b'=\bot$.
Assume this is not the case. Then, w.l.o.g.\ $a'\neq\bot$ and $\bot\in\ddata(a')$.
But this implies that $a'$ must be of the form $\Amb(\_,\_)$, which, however, is
impossible, since $a$ belongs to a regular type.
\end{proof}
\bigskip


\noindent
\textbf{Lemma \ref{lem:ssp}.}
Let $M$ be a closed program.
\begin{enumerate}
\item[(1)] $\ssp$ is deterministic (i.e., $M \ssp M'$ for at most one $M'$).
\item[(2)] $\ssp$ preserves the denotational semantics (i.e., $\val{M} = \val{M'}$
if $M \ssp M'$).
\item[(3)] $M$ is a $\ssp$-normal form iff $M$ is a  
w.h.n.f.  
\item[(4)] [Adequacy Lemma]
  If $\val{M} \ne \bot$, then there is a 
w.h.n.f.~$V$ 
s.t.\ $M \ssp^* V$.
\end{enumerate}

\begin{proof}
   (1) to (3) are easy.
   The proof of (4) is standard and is 
   as the proof of Lemma~33 in \cite{IFP} 
   for the case that $M$ begins with a constructor, 
   and an easy consequence of Lemma~32 in \cite{IFP}
   for the case that $M$ is a $\lambda$-abstraction. 
\end{proof}
\bigskip

\mps{\HT{Has to be completely rewritten. Do it later.}}
\noindent
\textbf{Theorem \ref{thm:data}} (Computational Adequacy: Soundness).  
%
For every computation  $M =  M_0 \newprintp M_1 \newprintp  \ldots$, 
$\sqcup_{i \in \NN} (M_i)_{D} \in \ddata(\val{M})$.

\begin{proof}

Set $P(d, a) \eqdef d = \sqcup_{i \in \NN} (M_i)_D$ for some
computation
$M_1 \newprintp M_2 \newprintp\ldots$ with $a = \val{M_1}$.
  We show $P(d, a) \to d \in \ddata(a)$, by coinduction.
  Therefore, we have to show $P(d, a) \to \Phi(P)(d, a)$ where
 the predicate $\Phi(X)$ is obtained by replacing $\ddata$ 
 with $X$ on the right hand side of the definition of $\ddata$.
 Assume $P(d,a)$, witnessed by the computation, that is, 
fair reduction sequence 
$M = M_0 \newprintp M_1 \newprintp \ldots$ with $a = \val{M}$
and $d = \sqcup_{i \in \NN} (M_i)_{D}$.
 We have to show that at least one of the following six conditions holds:
\begin{itemize}
\item[(1)] $a = \Amb(a', b') \land a' \ne \bot \land P(d,a')$
\item[(2)] $a = \Amb(a', b') \land b' \ne \bot \land P(d,b')$
\item[(3)] $a = \Amb(\bot, \bot) \land d = \bot$
\item[(4)] $ a =  C(\vec{a'}) \land d = C(\vec{d'}) \land \bigwedge_i P(d'_i,a'_i)$ for some
  $C \in \mathrm{C_d}$.
\item[(5)] $a = \Fun(f)  \land d = a$
\item[(6)] $a = d = \bot$
\end{itemize}
A computation $M = M_0 \newprintp M_1 \newprintp \ldots$ belongs to one the following categories:

Case a: All reductions $M_i \newprintp M_{i+1}$ are (p-i) derived from (c-i):
That is, $M_0 \ssp M_1 \ssp \ldots$. In this case,  $d = \bot$ and (6) holds by Lemma \ref{lem:ssp} (2) and (4).

Case b: For some $n$, $M_i \ssp M_{i+1}$  ($i< n$) and
$M_i \newprintc M_{i+1}$  ($n \leq i$) by (c-ii) and (c-ii'):
In this case, $M_i = \Amb(L_i, R_i)$ for $i \geq n$.
We have $L_i \ssp L_{i+1}$ and $R_i = R_{i+1}$ or 
$L_i = L_{i+1}$ and $R_i \ssp R_{i+1}$.  By fairness, both happens infinitely and therefore
(3) holds by Lemma \ref{lem:ssp} (2) and (4).

Case c:
$M_i \ssp M_{i+1}$ for $i< n$  and $M_i \newprintp M_{i+1}$ for $i \geq n$ by (p-iii):
(5) holds by Lemma \ref{lem:ssp}(2).

Case d:
$M_i \ssp M_{i+1}$ for $i< n$ and $M_i \newprintp M_{i+1}$ for $i \geq n$ by (p-ii):
$M_{i}$ has the form $C(N_{i,1},\ldots,N_{i,k})$ for $i \geq n$ 
and $N_{n,j} \newprintp N_{n+1,j} \newprintp \ldots$ are fair reductions.  In addition, 
$a = C(\val{N_{i,1}},\ldots,\val{N_{i,k}})$ and 
$d = C(\sqcup_{i \in \NN} ({N_{i,1}})_D,\ldots, \sqcup_{i \in \NN} (N_{i,k})_D)$.  Therefore, (4) holds.

Case e:
$M_i \ssp M_{i+1}$ for $i< n$, $M_n = \Amb(L,R)$, 
$M_{i} \newprintc M_{i+1}$ for $n \leq i< m$ by (c-ii) and (c-ii'),
$M_{m} \newprintc M_{m+1}$ by (c-iii) or (c-iii'):
$a = \Amb(a',b')$ with $a' = \val{L}$ and 
$b' = \val{R}$. $M_m = \Amb(L',R')$.
If (c-iii) is used, $M_{m+1} = L'$ with 
$a' = \val{L} = \val{L'} \neq \bot$.
Since the reduction sequence 
$M_{m+1} \newprintp M_{m+2} \newprintp \ldots$ is fair again,
$P(d,a')$ and hence (1) holds.  
Similarly, (2) holds for the case (c-iii') is used.

\end{proof}
\bigskip

\mps{\HT{need to check the proof}}
\noindent
\textbf{Theorem \ref{thm:dataconv}} (Computational Adequacy: Completeness).
If $M$ has a regular type, then
for every $d \in \ddata(\val{M})$, there is a computation 
$M = M_0 \newprintp M_1 \newprintp \ldots$ with
$d = \sqcup_{i \in \NN} ((M_i)_{D})$.
\begin{proof}
First, we observe that every 
regular type is semantically equal to  
a type 
\HT{$\Am(\rho)$ or $\rho$ where $\rho$ is neither a fixed point type nor 
of the form $\Am(\sigma)$.}
This follows from an easy modification of Lemma 7 in~\cite{IFP}.
The key to the proof of the theorem is the following

\begin{quote}
\emph{Claim.}
  Let $e$ be a 
compact
element of $D$.
If $M$ is a 
program 
with $M : \tau$ for some 
regular
type $\tau$,
$d \in \ddata(\val{M})$ and $e \sqsubseteq d$, then there exists 
$M'$ such that 
$M \newprintp^* M'$, $d \in \ddata(\val{M'})$, and 
$e \sqsubseteq M'_{D} \sqsubseteq d$.
\end{quote}

\noindent
\emph{Proof of the Claim.}
%
%
\HT{Induction} on the rank of $e$ 
where the \emph{rank} of a 
finite element of $D$ is defined as 
$\rk(\bot) = \rk(\Fun(f)) = 0$ and 
$\rk(C(\vec{a})) = 1+ \max (\rk(\vec{a}))$ (see also~\cite{IFP}). 
\mps{\HT{Now, $k = 0, 1$ so simple induction is enough.}}
%
\HT{
$\tval{\tau}{} = \tval{\Am^k(\rho)}{}$ 
where $k \in\{0, 1\}$ and
$\rho$ is neither a fixed point type nor of the form $\Am(\rho')$.}

\emph{Case $e = \bot$}. Then the assertion holds with $M' = M$, 
since clearly $M_{D} \sqsubseteq d$
for all $d \in\ddata(\val{M})$ (induction on $M$). 

\emph{Case $e = C(\vec{e'})$ with $C \ne \Amb$}. 
Then $d = C(\vec{d'})$ with $d_i' \sqsupseteq e_i'$.

If $k=0$, then $\tau$ is semantically equal to a type of the form 
$\one$,  $\rho_1+ \rho_2$ or $\rho_1\times\rho_2$,
and therefore $\val{M}$ has the form $C(\vec{a'})$. 
By the Adequacy Lemma (Lemma~\ref{lem:ssp}~(4)), 
$M \ssp^* C(\vec{M'})$ for some $\vec{M'}$
and $d \in \ddata(\val{M}) = \ddata(C(\vec{\val{M'}}))$.  
Therefore, by the definition of $\ddata$, $d_i' \in \ddata(\val{M'_i})$.  
Since the ranks of the $e_i'$ are smaller than that of $e$, 
by 
\HT{induction} hypothesis, there exists $\vec{M''}$ such that 
$M'_i \newprintp^* M''_i$, $d_i' \in \ddata(\val{M''_i})$ 
and $e_i' \sqsubseteq ({M_i}'')_{D} \sqsubseteq d_i'$.
Therefore, $C(\vec{M'}) \newprintp^* C(\vec{M''})$ by (p-ii),  
$d \in \ddata(\val{C(\vec{M''})})$, and
$e \sqsubseteq C(\vec{M''})_{D} \sqsubseteq d$.
Since $M \newprintp C(\vec{M'})$, we are done.

If \HT{$k = 1$}, 
then $\val{M}$ has the form $\Amb(a', b') $.
\HT{Since $d \sqsupseteq e \ne \bot$, }
$a' \ne \bot \land d \in \ddata(a')$
or $b' \ne \bot \land d \in \ddata(b')$.
By the Adequacy Lemma, $M \ssp^* \Amb(N_1, N_2)$.
If   $\defined{\val{N_1}} \land d \in \ddata(\val{N_1})$ then 
$N_1 \ssp^* K$ for some w.h.n.f.\ $K$ and therefore
$M \newprintc^* K$ by applying (c-i), (c-ii), and (c-iii),
and thus $M \newprintp^* K$ by (p-i).
Note that 
\HT{$K : \Am^{0}(\rho)$} and
$d \in \ddata(\val{K})$.  Therefore, 
there exists $K'$ such that $K \newprintp^* K'$,
$d \in \ddata(\val{K'})$, and $e \sqsubseteq K'_{D} \sqsubseteq d$.
Since $M \newprintp^* K'$, we have the result.

\emph{Case $e =\Fun(f)$}. 
If $k=0$, then $\val{M} = \Fun(f)$ and therefore $d = \Fun(f)$.
Furthermore, by the Adequacy Lemma, 
$M \ssp^*M'$ for some $M'$ in w.h.n.f. Since $\val{M'} = \val{M} = \Fun(f)$,
$M'$ is a $\lambda$-abstraction and hence $M'_{D} = \val{M'}$.
It follows that $e = M'_{D} = d$.
If \HT{$k = 1$} 
the same argument as in the case $e = C(\vec{e'})$ applies. 
This completes the proof of the Claim.

To prove the Theorem,  let $d_0 \sqsubseteq d_1 \sqsubseteq \ldots$ 
be an infinite sequence of 
  compact
  approximations of
  $d$ such that $d = \sqcup_i d_i$.  
  We construct a sequence $(M_i)_{i\in \NN}$ 
  such that $d \in \ddata(\val{M_i})$ and $M_i$ has a regular type as follows.
  Let $M_0 = M$.  By applying the Claim to 
  $d_i$, $d$ and $M_i$, we have $M_{i+1}$ such that $M_i \newprintp^* M_{i+1}$ 
  (hence, clearly, $M_{i+1}$ has a regular type as well),
  $d \in \ddata(\val{M_{i+1}})$, and 
  $d_i \sqsubseteq (M_{i+1})_{D} \sqsubseteq d$.
  By concatenating the reduction sequences, we have an infinite sequence
  $M = N_0 \newprintp N_1 \newprintp \ldots$ such that 
 $d = \sqcup_{i \in \NN} ((N_i)_{D})$.
\end{proof}
\bigskip

\noindent
\textbf{Corollary~\ref{cor:ddatabot}}.
For a program $M$ of regular type, the following 
are equivalent.
  \begin{enumerate}
   \item[(1)] One of the  computations of $M$ is productive.
   \item[(2)] All 
computations  of $M$ are productive.
   \item[(3)] $\val{M}$ is neither $\bot$ nor $\Amb(\bot,\bot)$.
   \end{enumerate}
 
\begin{proof}

Clearly, every program has a fair $\newprintp$-reduction sequences. Therefore, (2) implies
(1). 
Next, assume (1). Then, by Thm~\ref{thm:data}, $\ddata(\val{M})$ must contain a
non-bottom element. By Lemma~\ref{lem:ddatabot}, (3) holds.
Finally, if (3) holds, then by Lemma~\ref{lem:ddatabot}, $\bot\not\in\ddata(\val{M})$.
With \HT{Thm.~\ref{thm:data}} 
 it follows that every $\newprintp$-reduction sequences 
of $M$ must reduce to a deterministic w.h.n.f.
\end{proof}
\bigskip

\noindent
\textbf{Lemma~\ref{lem-restrict}}.
The rules for restriction and concurrency are realizable.
\begin{proof}\quad\\

\noindent\emph{Rest-intro}:
\ \raisebox{-0.2cm}{$\infer[\hbox{($A, B_0, B_1$ Harrop)}]{
        \ire{(\leftright\ b)}{\rt{A}{(B_0 \vee B_1)}}
}{
\ire{b}{(A \to (B_0 \vee B_1))} \ \ \     \reah({\neg A \to B_0 \wedge B_1})
}
$}.

$\ire{b}{(A \to (B_0 \vee B_1))}$ means $b:\tau(B_0\lor B_1)$ and
$\reah(A) \to \ire{b}{(B_0 \vee B_1)}$, and
$\reah(\neg A \to B_0 \wedge B_1) \equiv 
\neg \reah(A) \to \reah(B_0) \wedge \reah(B_1)$.
We claim that $\rt{A}{(B_0 \vee B_1)}$ 
is realized by 
$$\leftright\ b = 
\caseof{b} \{\Left(\_) \to \Left, \Right(\_) \to \Right\}.
$$
%
%
%
Assume $\re\,A$, that is, $\reah(A)$. 
Then $b$ realizes $B_0 \lor B_1$. 
Hence $b\in\{\Left,\Right\}$ and therefore $\defined{b}$
and thus $\leftright\ b \ne \bot$.
Now assume $\leftright\ b \ne \bot$.
We do a (classical) case analysis on whether or not $\reah(A)$ holds.
If $\reah(A)$, then $\ire{b}{(B_0 \vee B_1)}$.
If $\neg \reah(A)$, then $\reah(B_0)$ and $\reah(B_1)$. Hence,
$\Left$ and $\Right$ both realize $B_0 \lor B_1$.
Since $b:\tau(B_0\lor B_1)$ and $\defined{b}$,
$b \in \{\Left, \Right \}$.   Therefore, $\ire{b}{(B_0 \lor B_1)}$.

\medbreak\noindent
\emph{Rest-return}:\ \raisebox{-0.2cm}{$\infer{  
 \ire{a}{\rt{A}{B}}
}{
  \ire{a}{B}
}$} for strict $B$.

Since $B$ is strict, $\ire{a}{B}$ implies $a \ne \bot$.
Therefore, clearly $\ire{a}{\rt{A}{B}}$.

\medbreak
\noindent
\emph{Rest-bind}
:\ \raisebox{-0.2cm}{
\infer{
      \ire{(\strictapp{f}{a})}{\rt{A}{B'}}
}{
\ire{a}{ \rt{A}{B}}\ \ \          \ire{f}{( B \to (\rt{A}{B'}))}
}
} for strict $B$, $B'$, non-Harrop $B$.

We have $\forall c\,(\ire{c}{B} \to \ire{(f\,c)}{\rt{A}{B'})})$.
If $\re\,A$ then $\defined{a}$ and  $\ire{a}{B}$, and therefore $f\, a$ 
realizes $\rt{A}{B'}$.
Therefore $\defined{f\, a}$ because $\re\,A$.
Note that $\strictapp{f}{a} = f\, a$ because $\defined{a}$.
If $\defined{\strictapp{f}{a}}$, then $\defined{a}$.
Since $\defined{a}$ and $\ire{a}{\rt{A}{B}}$,  we have $\ire{a}{B}$.
Therefore,  $\strictapp{f}{a} = f\, a$ realizes $\rt{A}{B'}$.
If $B$ is Harrop,  
then
$a\,\seq\,b$ realizes $\rt{A}{B'}$ with a similar argument.

\medbreak
\noindent \emph{Rest-antimon}: \ \raisebox{-0.2cm}{
$\infer{
    \ire{a}{\rt{A'}{B}}
    }{
      \ire{}{(A' \to A)} \ \ \ \ire{a}{\rt{A}{B}  }}$} for strict $B$.

Clearly, $\ire{a}{\rt{A'}{B}}$
since $\re\,A'$ implies $\re\,A$.

\medbreak
\noindent \emph{Rest-mp}: \ \raisebox{-0.2cm}{
$  \infer{
    \ire{b}{B}
}{
\ire{b}{\rt{A}{B}} \ \ \    \ire{}{A}
}
$} for strict $B$.

Clear from the definition of $\ire{b}{\rt{A}{B}}$.

\medbreak
\noindent \emph{Rest-efq}: \ {
$\ire{\bot}{\rt{\False}{B}} $} for strict $B$.

Clear.

\medbreak
\noindent \emph{Rest-stab}: \ \raisebox{-0.2cm}{
$
\infer{
    \ire{b}{\rt{\neg\neg A}{B}}
    }{
    \ire{b}{\rt{A}{B}}}
$} for strict $B$.

We use classical logic.
If $\re\,\neg\neg A$, then $\neg\neg \re\, A$, hence
$\re\,A$. Thus $\defined{b}$. If $\defined{b}$, then 
$\ire{b}{B}$.

\medbreak
\noindent \emph{Conc-lem}: \ \raisebox{-0.2cm}{
$
  \infer{
  \ire{\Amb(a,b)}{\Set(B)}
}{
\ire{a}{\rt{A}{B}}   \ \ \ \     \ire{b}{\rt{\neg A}{B}}
}$} for strict $B$.

By classical logic
$\re\,A$ or $\neg(\re\,A)$ i.e.\ $\re\,(\neg A)$. In the first case
$a \ne \bot$  and in the second case $b \ne \bot$. Further, if $a \ne \bot$,
then $a$ is a realizer of 
$B$ because $a$ realizes $\rt{A}{B}$. Similarly for $b$.

\medbreak
\noindent \emph{Conc-return}: \ \raisebox{-0.2cm}{
$  \infer{\   \ire{\Amb(a,\bot)}{\Set(A)}
}{
\ire{a}{A}
}$}\ for strict $A$.

Clear.

\medbreak
\noindent \emph{Conc-return}: \ \raisebox{-0.2cm}{
$  \infer{\
\ire{\mapamb(f, c)}{\Set(B)}
}{
  \ire{f}{(A\to B)}\ \ \  \ire{c}{\Set(A) }
}
$}\ for strict $A$, $B$, non-Harrop $A$.

We show that
$$\mapamb(f, c) =  \caseof{c}
  \{\Amb(a,b) \to 
   \Amb(\strictapp{f}{a}, \strictapp{f}{b})\}$$
realizes $\Set(B)$.

If $\ire{\Amb(a, b)}{\Set(A)}$, then 
$\defined{a} \lor \defined{b}$.  If $\defined{a}$, then
$\strictapp{f}{a} = f\ a$ and $\ire{a}{A}$, therefore
$\ire{(\strictapp{f}{a})}{B}$.  Since $B$ is strict, we have
$\defined{\strictapp{f}{a}}$.   In the same way, 
if $\defined{b}$ then $\defined{\strictapp{f}{b}}$.
Therefore, we have $\defined{\strictapp{f}{a}} \lor \defined{\strictapp{f}{b}}$.  If $\defined{\strictapp{f}{a}}$, then $\defined{a}$ and thus
$\ire{(\strictapp{f}{a})}{B}$ as we have observed.
If  $A$ is Harrop,  then, since clearly $\re(A)$,
it is realized by $\Amb(f,\bot)$.

The cases where the conclusions of the rules are Harrop formulas are easy.
%
\end{proof}

\noindent
\textbf{Lemma~\ref{class-orelim-sleep}}.
$\CFP$ derives the following rules. 
%
\begin{itemize}
 \item[(1)]
\raisebox{-0.2cm}{
$  \infer[]{
   \ire{\Amb(\strictapp{\Left}{a},\strictapp{\Right}{b})}{\Set(B_0\lor B_1)}
 }{
    \ire{a}{\rt{A_0}{B_0}} \ \ \ \  \ire{b}{\rt{A_1}{B_1}}   \ \ \ \  \reah{(\neg\neg(A_0\lor A_1))}
 }$
}\\
\item[(2)]
\raisebox{-0.2cm}{
$
  \infer[]{
   \ire{\caseof{a} \{\Left(\_) \to \bot;\Right(b) \to b\}}{\rt{D\land\neg B}{C}}
 }{
    \ire{a}{\rt{D}{(B \lor C)}}
 }$}
\qquad {($C$ strict)}

\end{itemize}

\begin{proof}
  (1) By (Rest-mon) and (Rest-return),
  we have $(B|_A \land (B \to B')) \to B'|_A$.
Therefore, $B_0|_{A_0}$ implies $(B_0 \lor B_1)|_{A_0}$ and
$B_1|_{A_1}$ implies $(B_0 \lor B_1)|_{A_1}$.
On the other hand, 
using the rule (Rest-antimon) and (Rest-stab), 
$$\rt{A_0}{B} \land \rt{A_1}{B} \land \neg \neg{(A_0 \lor A_1)} \to \Set(B)$$ is derived from (Conc-lem).  Therefore,
$\Set(B_0 \lor B_1)$ is derived.

(2)
We first prove $B \lor C \to \rt{\neg B}{C}$.
Suppose that $B \lor C$.  If $B$, then $\rt{\neg B}{C}$ by (Rest-efq)  and (Rest-antimon). If $C$, then $\rt{\neg B}{C}$ by (Rest-return).

Now suppose $\rt{D}{(B \lor C)}$. 
By (Rest-antimon), we have $\rt{D \land \neg B}{(B \lor C)}$.
By (Rest-antimon) and $(B \lor C) \to \rt{\neg B}{C}$, we have $(B \lor C) \to 
\rt{D \land \neg B}{C}$.
Therefore, by (Rest-bind), we have $\rt{D \land \neg B}{C}$.
\end{proof}


\noindent
\textbf{Theorem~\ref{thm-faithfulness}} (Faithfulness).
If $a\in D$ realizes an admissible formula $A$, 
then all $d\in\ddata(a)$ realize $A^-$.

\begin{proof}
For any $\RCFP$ predicate $P$ whose last argument place ranges over $D$
we define a predicate $P'$ of the same arity
by
\[P'(\vec x,a) \eqdef 
                  \forall d\in\ddata(a)\, P(\vec x,d)\,.\]
We extend this to $\RCFP$ 
predicate substitutions by setting 
$\theta'(X) := (\theta(X))'$.

In the following, we consider only s.p. subexpression of the fixed 
closed admissible formula $A$ given in the theorem. 
The notion of a subexpression at F-position always 
refers to subexpressions of $A$. 

Let $P$ be a s.p.\ subexpression of $A$ all whose free predicate
variables are at non-F-position, and suppose that 
$\zeta, \theta$  are $\RCFP$ predicate substitutions such that
$\zeta \subseteq \theta'$ (that is, $\zeta(X) \subseteq (\theta(X))'$)
and the common domain of $\zeta$ and $\theta$ contains all free predicate 
variables of $\rea(P)$.
We show that 
\begin{equation}
\label{eq-faithfulness} 
\rea(P)\zeta\subseteq(\rea(P^-)\theta)' \, 
\end{equation}
This implies the theorem, since,
applying (\ref{eq-faithfulness}) to $A$
(with empty substitutions), we get
$\rea(A) \subseteq \rea(A^-)'$,
that is, all $d\in\ddata(a)$ realize $A^-$.

The proof of (\ref{eq-faithfulness}) is by 
structural induction on $P$.

If $P$ is Harrop, then $P$ contains no $\Set$ and 
no free predicate variables.
Therefore, $P^-$ is $P$ and (\ref{eq-faithfulness}) reduces to 
$ a =\Nil\land\reah(P) \to \forall d\in\ddata(a)\,(d=\Nil \land\reah(P))$, 
which is a triviality since $\ddata(\Nil)=\{\Nil\}$.

If $P$ is a subexpression at F-position, then it contains neither $\Set$
nor free predicate variables. Therefore, 
$\zeta$ and $\theta$ are empty, and (\ref{eq-faithfulness}) means
$\rea(P)\subseteq(\rea(P))'$. 
\mps{\UB{It might be worth checking this.}}
Since $P$ does not contain $\Set$,
the closed regular type $\tau(P)$ does not contain $\Am$ and 
from that it follows that
$\ddata(a)=\{a\}$ for all $a\in\val{\tau(P)}$. Therefore 
$(\rea(P))'=\rea(P)$.

Otherwise, we only look at the cases $A \to B$, 
$\Set(A)$, $X$ and $\munu\,\Phi$, 
since the other cases are easy or similar. 

If $P$ is an implication $B \to C$, then $B$ must be Harrop and
$C$ non-Harrop (since $P$ is non-Harrop and not at F-position),   
hence $(B \to C)^- = B \to C^-$.
$(\rea(B \to C)\zeta)(a)$ means 
$ a:\tau(C) \land  (\reah(B) \to (\ire{a}{C})\zeta)$,
and $(\rea{(B \to C^-)}\theta)'(a)$ means 
$\forall d \in\ddata(a)\,(d:\tau(C^-)\land(\reah(B) \to (\ire{d}{C^-})\theta))$.
Since, as one easily shows, $\ddata(a)\subseteq \tau(C^-)$
for $a: \tau(C)$,
the assertion holds by the structural induction hypothesis.

Assume $(\rea(\Set(B))\zeta)(\Amb(a,b))$ and $d\in\ddata(\Amb(a,b))$.
We have to show $(\rea(B^-)\theta)(d)$.
W.l.o.g.\ $a \ne \bot$, $d\in\ddata(a)$
and $(\rea(B)\zeta)(a)$.
By the structural induction hypothesis, $(\rea(B^{-})\theta)(d)$.

If $P$ is a predicate variable $X$, then (\ref{eq-faithfulness}) means
$\zeta(\reali{X})\subseteq (\theta(\reali{X}))'$, 
which holds by the assumption on $\zeta$ and $\theta$. 

For the case $\mu\,\Phi$, where $\Phi = \lambda X\, Q$, we have
$\rea(\mu\,\Phi)(\zeta)= \mu\,(\rea(\Phi)\,\zeta)$ and 
$\rea((\mu\,\Phi)^{-})\theta = \mu\,(\rea(\Phi^{-})\,\theta)$ 
(where $\Phi^- \eqdef \lambda X\,Q^-$).
Therefore, we have to show 
$\mu\,(\rea(\Phi)\,\zeta)\subseteq P'$ where 
$P := \mu\,(\rea(\Phi^{-})\,\theta)$. Hence, by s.p.\ induction, 
it suffices to show $(\rea(\Phi)\,\zeta)\,P' \subseteq P'$.
Since $\mu\,\Phi$ is at a non-F-position, every free
occurrence of $X$ in $Q$ is at a non-F-position. Therefore, the predicate 
$Q$ satisfies the premises of (\ref{eq-faithfulness}) and hence the structural 
induction hypothesis applies to it.
%
\begin{eqnarray*}
(\rea(\Phi)\,\zeta)\,P' 
&\equiv& 
\rea(Q)\,(\zeta[\tilde{X} := P'])\\
&\stackrel{\mathrm{i.h.}}{\subseteq}&
(\rea(Q^-)\,(\theta[\tilde{X} := P]))'\\
&\equiv&
((\rea(\Phi^{-})\,\theta)\,P)'
\equiv
P'
%
\end{eqnarray*}

The case $\nu\,\Phi$ can be obtained by dualization, using the operation
\[P^*(\vec x,d) \eqdef 
                  \exists a\,(d\in\ddata(a)\land P(\vec x,a))\]
which satisfies the adjunction
\[P \subseteq Q' \Leftrightarrow P^* \subseteq Q.\]
%
\end{proof}
\bigskip



\noindent
\textbf{Lemma~\ref{lem-gclosure}}.
Assume $\G(x)$. Then:
\begin{itemize}
\item[(1)] $\G(\tent(x))$, $\G(|x|)$, and $\G(-x)$;
\item[(2)] if $x \ge 0$, then $\G(2x-1)$ and $\G(1-x)$;
\item[(3)] if $|x|\le 1/2$, then $\G(2x)$.
\end{itemize}
\begin{proof}
This follows directly from the definition of $\G$ and 
elementary properties
of the tent function (recall $\tent(x)=1-|2x|$).
The extracted programs consist of simple manipulations of 
the given digit stream realizing $\G(x)$, concerning only its tail and  
first two digits. No nondeterminism is involved.

We only use the fixed point property of $\G$, namely that
for all $y$,  $\G(y)$ is equivalent to 
$|y|\le 1 \land \D(y) \land \G(\tent(y))$.

This equivalence has computational content which will show up in the programs
extracted from the proofs below: If it is used from left to right a 
stream (the realizer of $\G(y)$) is split into its head (realizer of $\D(y)$)
and tail (realizer of $\G(\tent(y))$). Using it
from right to left corresponds to the converse operation of adding 
a digit to a stream.

Proofs of (1-3):

(1) follows directly from the equivalence above and the fact that
$\tent(x)=\tent(-x)$. 
The three extracted programs are ${\sf f}_1(d:p) = p$, 
${\sf f}_2(d:p) = 1:p$, and 
${\sf f}_3(p) = \nh\ p$ where $\nh\ (d:p) = \nnot\ d: p$.

(2) Assume in addition $x\ge 0$. Then $\tent(x) = 1 - 2x$. 
Since $\G(\tent(x))$ and, by (1), $\G$ is closed under negation, we have
$\G(2x-1)$. 
Furthermore, since $0 \le 1-x\le 1$ we have $|1-x|\le 1$ and $\D(1-x)$.
Therefore, to establish $\G(1-x)$, it suffices to show $\G(\tent(1-x))$.
But $\tent(1-x) = 1-2(1-x)= 2x-1$ and we have shown $\G(2x-1)$ already.
The extracted programs are ${\sf f}_4(d:p) = \nh\ p$ and 
${\sf f}_5(d:p) = 1:\nh\ p$.

(3) Now assume $|x|\le 1/2$. Then $1-|2x| \ge 0$ and we have 
$\G(1-|2x|)$ (since $\G(x)$). Therefore, by (2), $\G(|2x|)$. 
Hence $\G(\tent(|2x|))$ and therefore also $\G(\tent(2x))$ since
$\tent(2x)=\tent(|2x|)$. Since $|x|\le 1/2$ implies $|2x|\le 1$ and 
$\G(x)$ implies $\D(x)$ and hence $\D(2x)$, it follows $\G(2x)$.
The extracted program is ${\sf f}_6(d:p) = d:1:\nh\ p$.

%
\end{proof}

%% file: appendix-program.tex
We explain the program and experiments of Section~7 in more detail.  The source code (GraySD.hs) is avaliable form the repository~\cite{githubUB}.


\subsection{Nondeterminism}
\label{sub-program-nondet}
Using the primitives of the Haskell libraries
\verb|Concurrent| and \verb|Exception|
we can implement nondeterministic choice through a program
\verb|ambL| that picks from a list nondeterministically
a terminating element (if exists). Although in our application 
we need only binary choice, we implement arbitrary finite choice
since it is technically more convenient and permits more applications,
e.g.\ Gaussian elimination (Section~8.4).
\begin{verbatim}
import Control.Concurrent
import Control.Exception

ambL :: [a] -> IO a
ambL xs = 
  do { m <- newEmptyMVar ; 
       acts <- sequence 
                 [ forkIO (do { y <- evaluate x ; putMVar m y }) 
                    | x <- xs ] ;
       z <- takeMVar m ;        
       x <- sequence_ (map killThread acts) ;
       seq x (return z)
     }
\end{verbatim}
Comments:
\begin{itemize}
\item \verb|newEmptyMVar| creates an empty mutable variable, 
\item \verb|forkIO| creates a thread,
\item \verb|evaluate| evaluates its argument to head normal form,  
\item \verb|putMVar m y| writes \verb|y| into the mutable variable \verb|m| 
       provided \verb|m| is empty,
\item the line \verb|seq x (return z)| makes sure that the threads are killed 
      before the final result \verb|z| is returned. 
\end{itemize}

\subsection{Extracting data}
\label{sub-program-data}
We define the domain $D$ (Section~2) and a program \verb|ed| on $D$
(`extract data') that, using \verb|ambL|,  nondeterministically selects 
a terminating argument of the constructor Amb.
\begin{verbatim}
data D = Nil | Le D | Ri D | Pair(D, D) | Fun(D -> D) | Amb(D, D)

ed :: D -> IO D
ed (Le d) = do { d' <- ed d ; return (Le d') }
ed (Ri d) = do { d' <- ed d ; return (Ri d') }
ed (Pair d e) = do { d' <- ed d ; e' <- ed e ; return (Pair d' e') }
ed (Amb a b) = do { c <- ambL [a,b] ; ed c } ;
ed d = return d
\end{verbatim}
\verb|ed| can be seen as an implementation of the operational semantics 
in Section~3.

\subsection{Gray code to Signed Digit Representation conversion}
\label{sub-program-gray}
We read-off the programs extracted in the Sections~5 and~6
to obtain the desired conversion function.
Note that this is nothing but a copy of the programs in those sections
  with type annotations for readability. The programs work without
  type annotation because Haskell infers their types.
The Haskell types contain only one type $D$.
Their types as CFP-programs are shown as comments in the code below. 

\paragraph{From Section~5.}
\begin{verbatim}
mapamb :: (D -> D)  -> D -> D  -- (B -> C) -> A(B) -> A(C)
         -- (A(B) is the type of Amb(a,b) where a,b are of type B)
mapamb = \f -> \c -> case c of {Amb(a,b) -> Amb(f $! a, f $! b)}

leftright :: D -> D   -- B + C -> B + C
leftright = \b ->  case b of {Le _ -> Le Nil; Ri _ -> Ri Nil}

conSD :: D -> D    -- 2 x 2 -> A(3) 
         -- (2 = 1+1, etc. where 1 is the unit type)
conSD = \c -> case c of {Pair(a, b) ->
      Amb(Le $! (leftright a), 
          Ri $! (case b of {Le _ -> bot; Ri _ -> Nil}))}
\end{verbatim}

\paragraph{From Section~6.}
\begin{verbatim}
gscomp :: D -> D  -- [2] -> A(3) 
gscomp (Pair(a, Pair(b, p))) = conSD (Pair(a, b))

onedigit :: D -> D -> D   -- [2] -> 3 -> 3 x [2]
onedigit  (Pair(a, Pair (b, p))) c = case c of {
       Le d -> case d of {
              Le _ -> Pair(Le(Le Nil), Pair(b,p));
              Ri _  -> Pair(Le(Ri Nil), Pair(notD b,p))
        };
        Ri _  -> Pair(Ri Nil, Pair(a, nhD p))}

notD :: D -> D   -- 2 -> 2
notD a = case a of {Le _ -> Ri Nil; Ri _ -> Le Nil}

nhD :: D -> D  -- [2] -> [2]
nhD (Pair (a, p)) = Pair (notD a, p)

s :: D -> D    -- [2] -> A(3 x [2])
s p = mapamb (onedigit p) (gscomp p)

mon :: (D -> D) -> D -> D   -- (B -> C) -> A(3 x B) -> A(3 x C)
mon f p = mapamb (mond f) p  
   where  mond f (Pair(a,t)) = Pair(a, f t)

gtos :: D -> D    -- [2]  -> [3] 
gtos = (mon gtos) . s
\end{verbatim}

\subsection{Gray code generation with delayed digits}
\label{sub-program-delay}
Recall that Gray code has the digits $1$ and $ -1$, modelled
as \verb|Ri Nil| and \verb|Le Nil|. A digit may as well be undefined ($\bot$) 
in which case it is modelled by a nonterminating computation (such as \verb|bot| below).
To exhibit the nondeterminism in our programs we generate digits with different
computation times. For example, \verb|graydigitToD 5| denotes the digit $1$ computed in
$500000$ steps, while \verb|graydigitToD 0| does not terminate and therefore 
denotes $\bot$.
\begin{verbatim}
delay :: Integer -> D
delay n  | n > 1     = delay (n-1)
         | n == 1    = Ri Nil
         | n == 0    = bot
         | n == (-1) = Le Nil
         | n < (-1)  = delay (n+1)
bot = bot

graydigitToD :: Integer -> D  
graydigitToD a | a == (-1) = Le Nil
               | a == 1    = Ri Nil
               | True      = delay (a*100000)
\end{verbatim}


The function \verb|grayToD| lifts this to Gray codes, that is, 
infinite sequences of partial Gray digits 
represented as elements of $D$:
\begin{verbatim}
-- list to Pairs
ltop :: [D] -> D
ltop = foldr (\x -> \y -> Pair(x,y)) Nil

grayToD :: [Integer] -> D
grayToD = ltop . (map graydigitToD)
\end{verbatim}
For example, \verb|grayToD (0:5:-3:[-1,-1..])| denotes the Gray code
$\bot:1:-1:-1,-1,\ldots$ where the first digit does not terminate, the second digit (1)
takes 500000 steps to compute and the third digit (-1) takes 300000 steps.
The remaining digits (all $-1$) take one step each.

\subsection{Truncating the input and printing the result}
\label{sub-io}
The program $\gtos$ transforms Gray code into signed digit representation, 
so both, input and output are infinite. To observe the computation, 
we truncate the input to some finite approximation which $\gtos$ will
map to some finite approximation of the output. This finite output
is a nondeterministic element of $D$ 
(i.e.~it may contain the constructor \verb|Amb|) 
from which we then can extract nondeterministically a deterministic 
data using the function \verb|ed| which can be printed.

In the following we define the truncation and the printing 
of deterministic finite data.

\paragraph{Truncating $d\in D$ at depth $n$.}
\begin{verbatim}
takeD :: Int -> D -> D
takeD n d | n > 0 = 
  case d of
    {
      Nil        -> Nil ;
      Le a       -> Le (takeD (n-1) a) ;
      Ri a       -> Ri (takeD (n-1) a) ;
      Pair(a, b) -> Pair (takeD (n-1) a, takeD (n-1) b) ;
      Amb(a,b)   -> Amb(takeD (n-1) a, takeD (n-1) b) ;
      Fun _      -> error "takeD _ (Fun _)" ;
    }
            | otherwise = Nil
\end{verbatim}

\paragraph{Showing a partial signed digit.}
\begin{verbatim}
dtosd :: D -> String
dtosd (Le (Ri Nil)) = " 1"
dtosd (Le (Le Nil)) = "-1"
dtosd (Ri Nil)      = " 0"
dtosd _             = " bot"
\end{verbatim}

\paragraph{Printing an element of $D$ that represents a finite deterministic 
signed digit stream.}
\begin{verbatim}
prints :: D -> IO ()
prints (Pair (d,e)) = putStr (dtosd d) >> prints e
prints Nil          = putStrLn ""
prints _ = error "prints: not a partial signed digit stream"
\end{verbatim}

\subsection{Experiments}
\label{sub-experiments}
As explained in Section~7,
there are three Gray codes of $0$:
\begin{eqnarray*}
a &=& \ \ \bot:1:-1,-1,-1,\ldots\\
b &=& \ \ \, 1:1:-1,-1,-1,\ldots\\
c &=& -1:1:-1,-1,-1,\ldots
\end{eqnarray*}
and the set of signed digit representations of $0$ is
$A \cup B \cup C$ where
\begin{eqnarray*}
A &=& \{0^\omega\}\\
B &=& \{0^k\!:\!1\!:\!(-1)^\omega \mid k \geq 0\}\\
C &=& \{0^k\!:\!(-1)\!:\!1^\omega \mid k \geq 0\}. 
\end{eqnarray*}
%
Our \verb|gtos| program 
nondeterministically produces 
an element of $A$ for input $a$,
an element of $A \cup B$ for input $b$, 
and an element of $A \cup C$ for input $c$.
As the following results show, 
the obtained value depends on 
the speed of computation of the individual Gray-digits.

\bigskip

Input $b$:
\begin{verbatim}
*GraySD> ed (takeD 50 (gtos (grayToD (1:1:[-1,-1..])))) >>= prints
 1-1-1-1-1-1-1-1-1-1-1-1-1-1-1-1-1-1-1-1-1-1-1-1 bot
\end{verbatim}

Input $c$:
\begin{verbatim}
*GraySD> ed (takeD 50 (gtos (grayToD (-1:1:[-1,-1..])))) >>= prints
-1 1 1 1 1 1 1 1 1 1 1 1 1 1 1 1 1 1 1 1 1 1 1 1 bot
\end{verbatim}

\bigskip

Input $a$ (demonstrating that the program can cope with an undefined digit):
\begin{verbatim}
*GraySD> ed (takeD 50 (gtos (grayToD (0:1:[-1,-1..])))) >>= prints
 0 0 0 0 0 0 0 0 0 0 0 0 0 0 0 0 0 0 0 0 0 0 0 0 bot
\end{verbatim}

\bigskip

Input $b$ with delayed first digit:
\begin{verbatim}
*GraySD> ed (takeD 50 (gtos (grayToD (2:1:[-1,-1..])))) >>= prints
 0 0 1-1-1-1-1-1-1-1-1-1-1-1-1-1-1-1-1-1-1-1-1-1 bot
\end{verbatim}

\bigskip

Same, but with more delayed first digit:
\begin{verbatim}
*GraySD> ed (takeD 50 (gtos (grayToD (10:1:[-1,-1..])))) >>= prints
 0 0 0 0 0 0 0 0 0 1-1-1-1-1-1-1-1-1-1-1-1-1-1-1 bot
\end{verbatim}

\bigskip

Input $1,1,1,\ldots$ which is the Gray code of  $2/3$:
\begin{verbatim}
*GraySD> ed (takeD 50 (gtos (grayToD ([1,1..])))) >>= prints
 1-1 1-1 1-1 1-1 1-1 1-1 1 0-1-1 1-1 1-1 1-1 1-1 bot
\end{verbatim}

\bigskip

Same, but with delayed first digit:
\begin{verbatim}
*GraySD> ed (takeD 50 (gtos (grayToD (2:[1,1..])))) >>= prints
 0 1 1-1 1-1 1-1 1-1 0 1 1-1 1-1 1-1 1-1 1-1 1-1 bot
\end{verbatim}

To see that the last two results are indeed approximations of signed digit 
representations of $2/3$, one observes that in the signed digit 
representation \verb|0 1| means the same as \verb|1-1| ($0+1/2 = 1-1/2$), 
so both results are equivalent to
\begin{verbatim}
 1-1 1-1 1-1 1-1 1-1 1-1 1-1 1-1 1-1 1-1 1-1 1-1 bot
\end{verbatim}
which denotes $2/3$.

\bigskip

Note that since our experiments use the nondeterministic program \verb|ed|,
the results obtained with a different computer may differ from the ones 
included here.
Our theoretical results ensure that, whatever the results are, 
they will be correct.